\documentclass[11pt]{article}
\usepackage[a4paper,margin=1in]{geometry}
\usepackage{color,graphicx,url,amsmath, amsthm, amsfonts,amssymb,bbm,bm}
\graphicspath{{./}{fig/}}
\definecolor{blueblack}{rgb}{0,0,.7}
\newcommand{\emphdef}[1]{%                                                      
  \textcolor{blueblack}{%                                                       
    \textbf{\emph{#1}}%                                                         
  }%                                                                            
}

\makeatletter
\newif\ifqed
\renewcommand{\qedhere}{\global\qedtrue\unskip\kern 10\p@ \begingroup
  \unitlength\p@ \linethickness{.4\p@} \framebox(6,6){} \endgroup}
\def\proof{\global\qedfalse\@ifnextchar[%]
  {\@xproof}{\@proof}}
\def\endproof{\unskip\kern 10\p@ \ifqed\else\begingroup \unitlength\p@
  \linethickness{.4\p@} \framebox(6,6){} \endgroup\fi \endtrivlist }
\def\@proof{\trivlist \item[\hskip\labelsep {\it\proofname}  ]
  \ignorespaces }
\def\@xproof[#1]{\trivlist \item[\hskip\labelsep{\it #1}]
  \ignorespaces }
\makeatother

\newcommand{\R}{\mathbb{R}}
\newcommand{\wC}{\widetilde{C}}
\newcommand{\wS}{\widetilde{S}}
\newcommand{\wt}{\widetilde}
\newcommand{\polylog}{\mathop{\text{polylog}}}
\newcommand{\Z}{\mathbb{Z}}
\newtheorem{theorem}{Theorem}
\newtheorem{lemma}[theorem]{Lemma}
\newtheorem{proposition}[theorem]{Proposition}
\newtheorem{corollary}[theorem]{Corollary}
\begin{document}

%%%%%%%%%%%%%%%%%%%%%%%%%%%%%%%%%%%%%%%%%%%%%%%%%%%%%%%%%%%%%%%%%%%%%%%%%%
\title{Testing Graph Isotopy on Surfaces%
  \thanks{This article is the full version of a conference article bearing
    the same title, which appeared in the Proceedings of the Twenty-Eighth
    Symposium on Computational Geometry, 2012.  We acknowledge partial
    support from the French ANR Blanc project ANR-12-BS02-005 (RDAM).}}%

\author{\'Eric Colin de Verdi\`ere%
  \thanks{Laboratoire d'informatique, \'Ecole normale sup\'erieure, CNRS,
    Paris, France. email: \protect\url{Eric.Colin.de.Verdiere@ens.fr}} % 
  \and Arnaud de Mesmay%
  \thanks{Laboratoire d'informatique, \'Ecole normale sup\'erieure, 
    Paris, France. email: \protect\url{Arnaud.de.Mesmay@ens.fr}}}%
\maketitle

\bigskip\bigskip

\begin{abstract}
  We investigate the following problem: Given two embeddings $G_1$
  and~$G_2$ of the same abstract graph~$G$ on an orientable surface~$S$,
  decide whether $G_1$ and~$G_2$ are isotopic; in other words, whether
  there exists a continuous family of embeddings between $G_1$ and~$G_2$.

  We provide efficient algorithms to solve this problem in two models.  In
  the first model, the input consists of the arrangement of $G_1$
  (resp.,~$G_2$) with a fixed graph cellularly embedded on~$S$; our
  algorithm is linear in the input complexity, and thus, optimal.  In the
  second model, $G_1$ and~$G_2$ are piecewise-linear embeddings in the
  plane minus a finite set of points; our algorithm runs in $O(n^{3/2}\log
  n)$ time, where $n$ is the complexity of the input.

  The graph isotopy problem is a natural variation of the homotopy problem
  for closed curves on surfaces and on the punctured plane, for which
  algorithms have been given by various authors; we use some of these
  algorithms as a subroutine.

  As a by-product, we reprove the following mathematical characterization,
  first observed by Ladegaillerie (1984): Two graph embeddings are isotopic
  if and only if they are homotopic and congruent by an oriented
  homeomorphism.

\end{abstract}

%%%%%%%%%%%%%%%%%%%%%%%%%%%%%%%%%%%%%%%%%%%%%%%%%%%%%%%%%%%%%%%%%%%%%%%%%%
\section{Introduction}

Deciding whether a curve in a space can be deformed to another one is a
fundamental topological task.  Since the problem is undecidable in general,
even in two-dimensional simplicial complexes and
four-manifolds~\cite[p.~242--247]{s-ctcgt-93}, one has to put restrictions
on the ambient space.  The case of surfaces has been much investigated
recently (see references below), since it is at the same time non-trivial,
interesting, and tractable, and since the underlying mathematics, notably
in the field of combinatorial group theory, are well-understood.

This paper also studies deformations on surfaces, but instead of curves, we
consider graphs drawn on surfaces: Given two embeddings (crossing-free
drawings) of a graph~$G$ on a surface, can we deform one continuously to
the other without introducing intersections between edges during the
process?

As a motivating special case, consider a finite set of obstacle points~$P$
in the plane and a graph~$G$ embedded in~$\R^2\setminus P$ in two different
ways, $G_1$ and~$G_2$.  Does there exist a ``morph'' between $G_1$
and~$G_2$ (possibly moving the vertices and bending the edges) that avoids
passing over any obstacle?  To the best of our knowledge, no algorithm is
known for this purpose.  This is relevant for morphing applications: To
compute a morphing between two images, it is helpful to first build a
deformation between compatible graphs representing the most salient
features of the images.  In such applications, it is sometimes desirable to
add some topological requirements on the morphing, e.g., to force some area
of the deforming image to always cover a fixed point of the plane during
the deformation.  Such requirements can be encoded using obstacle points,
since a face of the graph containing an obstacle point has to contain it
during the whole deformation.

Another motivation for this problem comes from geographic information
systems and map simplification.  When simplifying a road network, it is
crucial that the features of the map (cities, mountains) stay on the same
side of the roads.  This can be tested by considering these features as
obstacle points and testing whether the road networks obtained before and
after straightening are isotopic; see, e.g., Cabello et
al.~\cite{clms-thpp-04}.

More generally, assume that we have a triangulated surface in~$\R^3$, and
two embeddings $G_1$ and~$G_2$ of the same graph~$G$ on that surface (not
necessarily on the skeleton of the triangulation). Each graph $G_i$ is encoded by its
combinatorial arrangement with the triangulation. Can we continuously move
$G_1$ to~$G_2$?  In this setting, the graphs $G_1$ and~$G_2$ might
represent textures on the surface, and the question is whether one can
continuously move one texture so that it coincides with the other.

%%%%%%%%%%%%%%%%%%%%%%%%%%%%%%%%%%%%%%%%%%%%%%%%%%%%%%%%%%%%%%%%%%%%%%%%%%
\subsection{Related Work}

A \emph{homotopy} between two paths is a continuous family of paths with
the same endpoints between them.  A \emph{homotopy} between two cycles
(closed curves) is a continuous family of cycles between
them.\footnote{This is sometimes referred to as \textit{free} homotopy in
  textbooks, in opposition to \textit{fixed point homotopy}, which we will
  only use in Section~\ref{S:fixed}.} %
Dey and Guha~\cite{dg-tcs-99}, and later Lazarus and
Rivaud~\cite{lr-hts-12} as well as Erickson and
Whittlesey~\cite{ew-tcsr-13} study the problem of deciding whether two
paths or cycles on a surface~$S$ are homotopic.  Both problems can be
solved in optimal linear time if the input curves are represented as walks
in a graph embedded on~$S$.  Cabello et al.~\cite{clms-thpp-04} give
efficient algorithms for testing homotopy of paths in the special case
where the surface~$S$ is a punctured plane (a plane minus a finite set of
obstacle points) and the input paths are represented by polygonal paths in
the plane.

A related problem, to which a large body of research is devoted, is that of
computing shortest homotopic paths or cycles.  It is first studied by
Hershberger and Snoeyink~\cite{hs-cmlpg-94} for a triangulated surface
where the vertices lie on the boundary, and revisited by Efrat et
al.~\cite{ekl-chspe-06} and Bespamyatnikh~\cite{b-chspp-03} for paths in a
punctured plane.  These algorithms, in particular, allow to decide whether
two paths are homotopic, and have a better complexity than the algorithm by
Cabello et al.~\cite{clms-thpp-04} in some cases.  Results on the shortest
homotopic path/cycle problem are also known for combinatorial
surfaces~\cite{cl-oslos-05,cl-opdsh-07}: It is solvable in polynomial
time~\cite{ce-tnpcs-10}.

An \emph{isotopy} between two simple curves (paths or cycles without
self-intersections) is a homotopy that keeps the deforming curve simple at
each step.  If two simple curves are isotopic, they are homotopic.  The
converse does not hold in general~\cite{f-hai-66}, though it holds for
simple paths with endpoints on the boundary and for simple cycles not
bounding a disk~\cite{e-c2mi-66}; in these cases, it is equivalent to test
homotopy or isotopy.  More generally, an \emph{isotopy with fixed vertices}
between two graph embeddings is a continuous family of graph embeddings
between them, where the vertices are not allowed to move.  It is known that
one can compute shortest graph embeddings with fixed vertices in polynomial
time~\cite{c-rcds-03,ce-tnpcs-10}, and similar techniques allow to test
whether two graph embeddings are isotopic with fixed vertices in polynomial
time.

Two homeomorphisms from $S$ to~$S$ are \emph{isotopic} if there is a
continuous family of homeomorphisms connecting them.  The \emph{mapping
  class group} of a surface (without boundary) $S$ is, roughly, the set of
isotopy classes of all orientation-preserving homeomorphisms from $S$
to~$S$.  See, e.g., Farb and Margalit~\cite{fm-pmcg-11} for a recent and
exhaustive survey on this topic.  Although we use little of this vast
theory, it is quite connected to our problem: If $G_1$ and~$G_2$ are
cellularly embedded on~$S$, a homeomorphism of~$S$ that maps $G_1$ to~$G_2$
represents a unique element of the mapping class group, and our problem
amounts to testing whether this element is the identity. Hence it is
closely related to the \emph{word problem} in mapping class groups, which
can be solved in quadratic time~\cite{h-cwpbg-00,m-mcga-95}. However, these
algorithms take as input the collection of Dehn twists corresponding to a
mapping class, which is an input incomparable to ours. Note
that if $S$ is a $n$-punctured sphere, the mapping class group of $S$ is
the pure braid group with $n$ strands, which has garnered considerable
algorithmic attention in recent years, in particular due to its possible
applications to cryptography~\cite{klchkp-npkcu-00}.

%%%%%%%%%%%%%%%%%%%%%%%%%%%%%%%%%%%%%%%%%%%%%%%%%%%%%%%%%%%%%%%%%%%%%%%%%%
\subsection{Our Results}

In this paper, we study the problem of deciding whether two embeddings
$G_1$ and~$G_2$ of the same abstract graph~$G$ on a surface~$S$ are
isotopic; in other words, whether there exists a continuous family of
embeddings of~$G$ (possibly moving the vertices of~$G$) connecting $G_1$
with~$G_2$.

In more detail, the input to the algorithm is a description of the
surface~$S$, the graph~$G$, and the graph embeddings $G_1$ and~$G_2$.  All
surfaces are assumed to be compact, connected, and orientable; they may
have boundary.  An embedding maps each vertex (or edge, or halfedge) of~$G$
to the corresponding feature on the surface.  In particular, the
correspondence between the vertices (or edges, or halfedges) of $G_1$
and~$G_2$ is given.

Our algorithmic results come in two flavors, depending on the model used.
In the general model, $G_1$ and~$G_2$ are embeddings of~$G$ on an arbitrary
surface~$S$. To represent them, we use a model similar to the
\emph{cross-metric surface model}~\cite{ce-tnpcs-10}: A graph embedding is
represented by the intersections it forms with a given cellularly embedded
graph. More specifically, we assume that $S$ has a fixed graph~$H$
cellularly embedded on it (i.e., each face of~$H$ is homeomorphic to a
disk; for example, $H$ may be a triangulation of~$S$); the input to the
algorithm is the combinatorial map of the arrangement of $G_1$ with~$H$
(resp., $G_2$ with~$H$)\footnote{i.e., the graph obtained by overlaying
  $G_i$ and $H$.}.  We emphasize that the input does not consider the
crossings betweeen $G_1$ and~$G_2$.  We give a linear-time (and thus
optimal) algorithm to decide whether $G_1$ and~$G_2$ are isotopic:
\begin{theorem}\label{T:main-surf}
  Let $S$ be an orientable surface, possibly with boundary.  Let $H$ be a
  fixed graph cellularly embedded on~$S$.  Let $G_1$ and~$G_2$ be two graph
  embeddings of the same graph~$G$ on~$S$, each in general position with
  respect to~$H$.  Given the combinatorial map of the arrangement of~$G_1$
  with~$H$ (resp., $G_2$ with~$H$), of complexity $k_1$ (resp.,~$k_2$), we
  can determine whether $G_1$ and~$G_2$ are isotopic in $O(k_1+k_2)$ time.
\end{theorem}
Let us emphasize that the surface~$S$ is not fixed in this result; the
constant in the $O(\cdot)$ notation does not depend on~$S$.

We also study the complexity of the problem in the case where $S$ is the
plane minus a finite set~$P$ of obstacle points, and $G_1$ and~$G_2$ are
piecewise-linear graph embeddings of~$G$ in~$\R^2\setminus P$.  In this
case, the input is the point set~$P$ together with the embeddings $G_1$
and~$G_2$, where each edge of each embedding is represented as a polygonal
path.  (Again, the embeddings $G_1$ and~$G_2$ may intersect arbitrarily.)
In this setting, we prove the following result.
\begin{theorem}\label{T:main-plane}
  Let $P$ be a set of $p$ points in the plane, and let $G_1$ and~$G_2$ be
  two piecewise-linear graph embeddings of the same graph~$G$ in
  $\R^2\setminus P$, of complexities (number of segments) $k_1$ and~$k_2$
  respectively.  We can determine whether $G_1$ and~$G_2$ are isotopic in
  $\R^2\setminus P$ in time $O(n^{3/2}\log n)$ time, where $n$ is the total
  size of the input.  In more detail, the running time is, for
  any~$\varepsilon>0$,
 \[ 
    O\Big((k_1+p)\log(k_1+p)+(k_2+p)\log(k_2+p)\ +
        \min\Big\{(k_1+k_2)p\ ,\
      p^{1+\varepsilon}+(k_1+k_2)\sqrt{p}\log p\Big\}\Big).
 \]
\end{theorem}

Let us emphasize that the isotopy is a continuous family of topological
embeddings; one may assume that all these embeddings are piecewise-linear,
but we claim no upper bound on their complexities.

%%%%%%%%%%%%%%%%%%%%%%%%%%%%%%%%%%%%%%%%%%%%%%%%%%%%%%%%%%%%%%%%%%%%%%%%%%
\subsection{Overview of the Techniques}

If two graph embeddings $G_1$ and~$G_2$ of the same graph~$G$ are isotopic,
then clearly (1) there is an oriented homeomorphism of the surface that
maps $G_1$ to~$G_2$%
\footnote{assuming that $G_1$ and~$G_2$ are in the interior of~$S$, which
  is true by our general position assumption.}%
; and (2) if $\gamma$ is a cycle in~$G$ (possibly with
repeated vertices and edges), then its images in~$G_1$ and~$G_2$ are homotopic.

It was shown by Ladegaillerie~\cite{l-cip1c-84} that such necessary
conditions are, in fact, sufficient.  However, the second condition is not
algorithmic, since there are infinitely many cycles in~$G$.  A close
inspection of Ladegaillerie's proof reveals that $O(g+b)$ pairs of cycles
need to be tested for homotopy, where $g$ and $b$ denote the genus and
number of boundary components of the surface.  However, in Ladegaillerie's
construction, the complexity of the family of cycles is not explicitly
given.  With some work, and using some of our techniques, one might be able
to obtain an explicit algorithm using his construction, but the running
time would certainly be larger, by at least an additional $O(g)$~factor
(since Ladegaillerie's decomposition contains a pants decomposition), if
not more.

Using a very different method, we reprove Ladegaillerie's characterization
in a strengthened form (Theorem~\ref{T:ladeg} below): We provide an
explicit set of cycles~$\Lambda$ in~$G$ of linear overall complexity that
have to be tested for homotopy.  Our algorithmic results follow, since one
can perform efficiently the homeomorphism test (this essentially amounts to
checking equality of two combinatorial maps, see Section~\ref{S:combmaps}), as well as the test of
homotopy between two given cycles (using known algorithms or variants of
them~\cite{dg-tcs-99,lr-hts-12,clms-thpp-04}).

We note that it is not straightforward to find a suitable set~$\Lambda$
in~$G$ of overall linear complexity satisfying the above condition.  In
particular, a natural candidate for~$\Lambda$ would be the set of all facial
cycles in~$G_1$ (and thus in~$G_2)$.  However, Figure~\ref{fig:contreex}
shows that the condition is not fulfilled, even in the case where the
surface is the sphere with four punctures.

\begin{figure}[htb]
  \centering
  \def\svgwidth{10cm}
\begingroup
  \makeatletter
  \providecommand\color[2][]{%
    \errmessage{(Inkscape) Color is used for the text in Inkscape, but the package 'color.sty' is not loaded}
    \renewcommand\color[2][]{}%
  }
  \providecommand\transparent[1]{%
    \errmessage{(Inkscape) Transparency is used (non-zero) for the text in Inkscape, but the package 'transparent.sty' is not loaded}
    \renewcommand\transparent[1]{}%
  }
  \providecommand\rotatebox[2]{#2}
  \ifx\svgwidth\undefined
    \setlength{\unitlength}{2326.17722882pt}
  \else
    \setlength{\unitlength}{\svgwidth}
  \fi
  \global\let\svgwidth\undefined
  \makeatother
  \begin{picture}(1,0.34951523)%
    \put(0,0){\includegraphics[width=\unitlength]{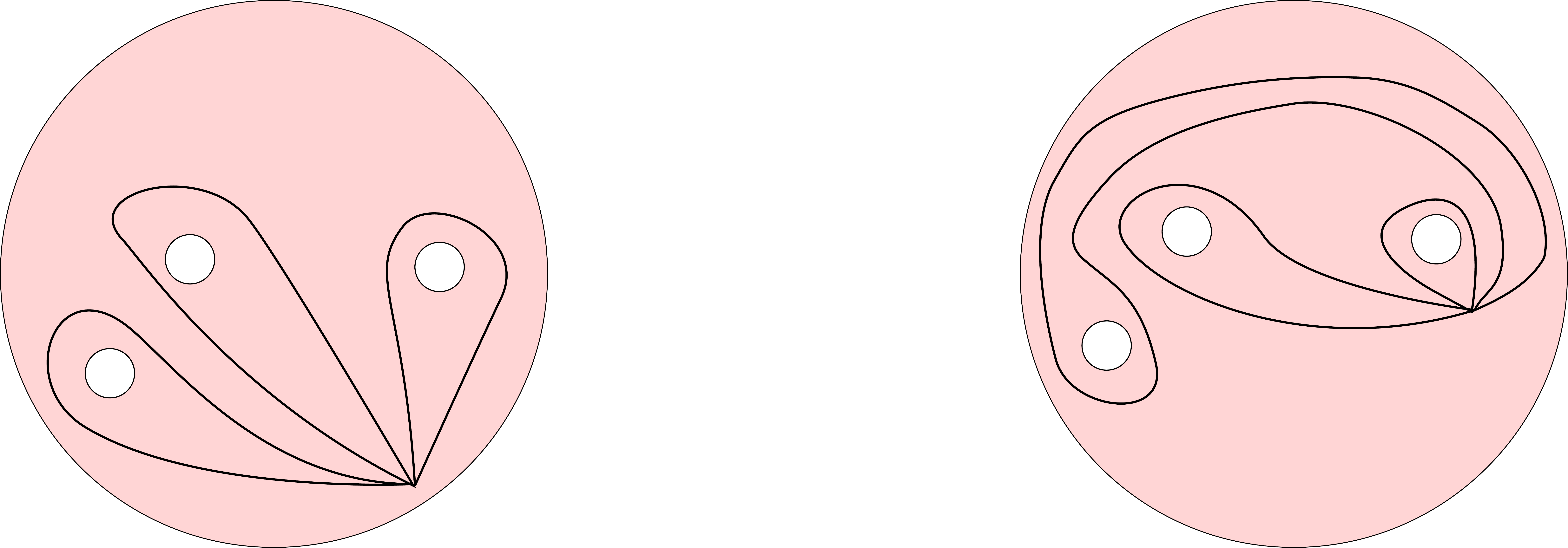}}%
    \put(0.03964688,0.163455){\color[rgb]{0,0,0}\makebox(0,0)[lb]{\smash{$e_1$}}}%
    \put(0.1359422,0.23911562){\color[rgb]{0,0,0}\makebox(0,0)[lb]{\smash{$e_2$}}}%
    \put(0.26171569,0.22339393){\color[rgb]{0,0,0}\makebox(0,0)[lb]{\smash{$e_3$}}}%
    \put(0.70290551,0.06912489){\color[rgb]{0,0,0}\makebox(0,0)[lb]{\smash{$e_1$}}}%
    \put(0.77070531,0.24304604){\color[rgb]{0,0,0}\makebox(0,0)[lb]{\smash{$e_2$}}}%
    \put(0.86405283,0.23321999){\color[rgb]{0,0,0}\makebox(0,0)[lb]{\smash{$e_3$}}}%
  \end{picture}%
\endgroup
  \caption{Two embeddings $G_1$ and~$G_2$ of a one-vertex graph with three
    loop edges on the sphere with four punctures.  These two embeddings are
    not isotopic, although there exists an oriented homeomorphism mapping
    one to the other, and the four cycles following the boundaries of the
    faces are homotopic in $G_1$ and~$G_2$}
  \label{fig:contreex}
\end{figure}
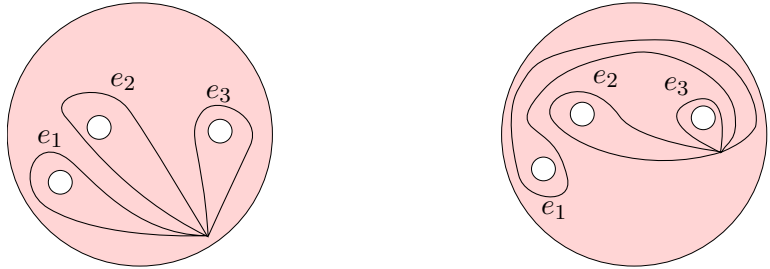

The paper is organized as follows.  After some preliminaries
(Section~\ref{S:prelim}), we characterize the existence of an isotopy
between two \emph{stable} arrangements of cycles, containing no disk with
less than four vertices on its boundary (Section~\ref{S:stable}).  We then
prove the aforementioned strengthened form of Ladegaillerie's result
(Section~\ref{S:ladeg}); the idea for the proof is to compute from~$G_1$
and~$G_2$ two stable arrangements of cycles in their tubular neighborhoods,
which are isotopic if and only if $G_1$ and~$G_2$ are isotopic, and to
apply the characterization for stable cycles.  We then deduce our main
computational results in Section~\ref{S:main-thms}.  In
Section~\ref{S:fixed}, we briefly indicate how our algorithms and results
extend to the case where we require some vertices of the graph to be fixed
throughout the isotopy.

%%%%%%%%%%%%%%%%%%%%%%%%%%%%%%%%%%%%%%%%%%%%%%%%%%%%%%%%%%%%%%%%%%%%%%%%%%
\section{Preliminaries}\label{S:prelim}

Here, we review some topological notions.  We refer the reader to
Hatcher~\cite{h-at-02}, Stillwell~\cite{s-ctcgt-93}, or any previous paper
on related topics~\cite{ce-tnpcs-10} for standard background in algebraic
topology used in the present paper, and Farb and Margalit~\cite{fm-pmcg-11}
for the connection between covering spaces, hyperbolic geometry, and
surface homeomorphisms.

%%%%%%%%%%%%%%%%%%%%%%%%%%%%%%%%%%%%%%%%%%%%%%%%%%%%%%%%%%%%%%%%%%%%%%%%%%
\subsection{Background}

%%%%%%%%%%%%%%%%%%%%%%%%%%%%%%%%%%%%%%%%%%%%%%%%%%%%%%%%%%%%%%%%%%%%%%%%%%
\paragraph*{Curves and graphs on surfaces. }

\ Henceforth, $S$ is a compact, connected, orientable surface with genus~$g$
and $b$ boundary components; its \emphdef{Euler characteristic} is
$\chi(S)=2-2g-b$.  A \emphdef{path}, respectively a \emphdef{cycle}, on a
surface $S$ is a continuous map $p:[0,1] \rightarrow S$, respectively a
continuous map $\gamma:S^1 \rightarrow S$, where $S^1$ is the unit
circle. An \emphdef{arc} is a path intersecting the boundary of~$S$ exactly
at its endpoints.  A curve denotes a path or a cycle. Paths and cycles are
\emphdef{simple} if they are one-to-one.  We emphasize that, according to
our definition, a cycle in a graph is a closed walk \emph{possibly with
  repeated vertices and edges} (in contrast to the graph-theoretic notion).

An \emphdef{embedding} of a graph~$G$ on a surface~$S$ is, informally, a
crossing-free drawing of~$G$ on~$S$.  A graph embedding is
\emphdef{cellular} if its \emphdef{faces}, namely, the connected components
of the complement of the image of the graph, are homeomorphic to open
disks.  If~$G$ is embedded cellularly on a surface $S$ and if we denote by
$v, e$, and $f$ the respective numbers of its vertices, edges and faces, we
have $v-e+f=\chi(S)=2-2g-b$.

A (finite) family of cycles on~$S$ is \emphdef{in general position} if the
cycles are in the interior of~$S$, there are finitely many
(self-)intersection points, and each intersection is a transverse crossing
between exactly two pieces of cycles.  Similarly, a (finite) family of
graph embeddings on~$S$ is \emphdef{in general position} if all embeddings
are in the interior of~$S$, there are finitely many intersection points
between two different embeddings, and each intersection is a transverse
crossing of the interiors of exactly two edges. Classical approximation
techniques, see for example Epstein~\cite[Appendix]{e-c2mi-66}, allow us to
approximate every edge in a graph embedding by a piecewise linear edge
using an ambient isotopy.  In particular, all graph embeddings in this
paper can be assumed to be piecewise-linear.  By doing small perturbations
if necessary, this allows us to assume that, moreover, all the graph
embeddings we consider are in general position, and we will always make
this assumption unless stated otherwise.

The simultaneous drawing of a family of cycles or of graph embeddings in
general position on~$S$ gives a set of vertices (original vertices or
intersection points) and edges (connecting two such vertices); this graph
is called the \emphdef{arrangement} of the family.  (Here we also allow
edges that are simple cycles without any vertex on them.)

%%%%%%%%%%%%%%%%%%%%%%%%%%%%%%%%%%%%%%%%%%%%%%%%%%%%%%%%%%%%%%%%%%%%%%%%%%
\paragraph*{Homeomorphisms and isotopies.}

\ At several occasions, we will consider a homeomorphism~$h$ that
\emphdef{maps} a cycle~$c:S^1\to S$ into another one, $c':S^1\to S$.
Unless stated otherwise, this expression means that $h\circ c=c'$, namely,
$h$ \emph{pointwise} maps $c$ to~$c'$.  However, we will sometimes need
weaker concepts.  We say that $h$ \emphdef{maps} $c$ to~$c'$ not
necessarily pointwise, but only \emphdef{as sets}, if they do so up to
reparameterization; namely, if there is a homeomorphism $\varphi:S^1\to
S^1$ such that $h\circ c\circ\varphi=c'$.  If furthermore the homeomorphism
$\varphi$ is increasing (intuitively, $h$ maps $c$ to the cycle~$c'$ not
necessarily pointwise, but the orientations of $h(c)$ and~$c'$ are the
same), we say that $h$ \emph{maps} $c$ to~$c'$ not necessarily pointwise,
but \emphdef{preserving the orientations} of the cycles.

Often we will also have an abstract~$G$ and two embeddings $G_1$ and~$G_2$
of~$G$ on a surface~$S$; we say that a homeomorphism~$h$ \emphdef{maps}
$G_1$ to~$G_2$ if it maps each edge of~$G_1$ (pointwise) to the
corresponding edge of~$G_2$.

Two embeddings $G_0$ and $G_1$ of the same abstract graph~$G$ on~$S$ are
\emphdef{isotopic} if there is a continuous family of embeddings
$(G_t)_{t\in[0,1]}$ between $G_0$ and~$G_1$.  In more detail, the data of
an embedding is given by the choice of a point for each vertex and a path
connecting the appropriate vertices for each edge (with some conditions
asserting that no crossing occurs); a family $(G_t)$ of embeddings is
continuous if all these maps vary continuously over~$t$.  (The vertices
may, in particular, move.) 

Two homeomorphisms $h_0$ and $h_1$ from~$S$ to~$S$ are \emphdef{isotopic}
if there is a continuous family of homeomorphisms $(h_t)_{t\in[0,1]}$
connecting them; the family $(h_t)$ is called an \emphdef{(ambient)
  isotopy} between them. By a common abuse of language, a homeomorphism
isotopic to the identity is also called an \emphdef{(ambient) isotopy}.
For $A \subseteq S$, we say that a homeomorphism~$h$ is an ambient isotopy
\emphdef{relatively to~$\bm{A}$} if there is a continuous family of
homeomorphisms between $h$ and the identity such that each homeomorphism is
the identity on~$A$.

It holds that two smooth embeddings of the cycle $S^1$ on~$S$ are isotopic
if and only if there exists an ambient isotopy of~$S$ mapping one to the
other (provided the two embeddings are in the interior of~$S$); this
follows from tools in differential topology (e.g., vector fields); see for
instance Farb and Margalit~\cite[Proposition~1.11]{fm-pmcg-11} or
Hirsch~\cite[Theorem~1.3]{h-dt-94}. In the case of (piecewise-linear) graph
embeddings, the same tools might yield the same result, but we did not find
any proof in the literature.  It is not an issue though, since our proof
also implies the statements of Theorems \ref{T:main-surf}
and~\ref{T:main-plane} in which graph isotopy is replaced with ambient
isotopy (see Corollaries \ref{C:ext-cont} and~\ref{C:ext-pl}).

We will need the following basic lemma.
\begin{lemma}[{Alexander's Lemma; see, e.g., Farb and Margalit~\cite[Lemma
    2.1]{fm-pmcg-11}}]
  \label{L:Alexander}
  Let $D$ be a disk and $h: D \rightarrow D$ be a homeomorphism
  fixed on the boundary of $D$. Then $h$ is an ambient isotopy relatively
  to the boundary of $D$.
\end{lemma}

\begin{proof}
  We can assume that $D$ is the unit closed disk in the plane.  The
  continuous family of embeddings between~$h$ and the identity can be
  explicitly defined by
  \[ F(x,t)= \left\{\begin{array}{ll}(1-t)h(\frac{x}{1-t}) & \text{if
      }0\leq |x| \leq 1-t\\x & \text{if }1-t \leq |x| \leq
      1 \end{array}\right. \]%
  for $0 \leq t <1$, and $F(x,1)=x$ for each $x\in D$.
\end{proof}

%%%%%%%%%%%%%%%%%%%%%%%%%%%%%%%%%%%%%%%%%%%%%%%%%%%%%%%%%%%%%%%%%%%%%%%%%%
\paragraph*{Covering spaces and hyperbolic geometry.}

\ If $S$ has a negative Euler characteristic and has no boundary, it can be
provided with a hyperbolic metric; this naturally induces a hyperbolic
metric on its universal cover $\widetilde{S}$, which is then isometric to
the open hyperbolic disk~\emphdef{$\bm{D_2}$}. This open disk is
classically compactified with a boundary $\partial D_2$ \cite[Chapter
1]{fm-pmcg-11}, and a non-contractible cycle on $S$ lifts into an arc in
$D_2\cup \partial D_2$. We call its intersections with the boundary the
\emphdef{endpoints} of the lift or, according to the orientation of the
lift, its \emphdef{source} and its \emphdef{target}. Two such lifts in
$D_2$ have common endpoints if and only if they stay at a bounded distance
from each other.  If $S$ has boundary, $\widetilde{S}$ is isometric to a
totally geodesic subspace of~$D_2$.

It is folklore~\cite[Lemma~1.6.5]{b-gscrs-92} that if two cycles on a
surface $S$ are homotopic, they have lifts in $D_2$ with the same
endpoints. We also recall that on a hyperbolic surface, every cycle is
homotopic to a unique geodesic~\cite[Proposition 1.3]{fm-pmcg-11}.

%%%%%%%%%%%%%%%%%%%%%%%%%%%%%%%%%%%%%%%%%%%%%%%%%%%%%%%%%%%%%%%%%%%%%%%%%%
\subsection{Combinatorial Maps for Non-Cellular Embeddings}\label{S:combmaps}

To each embedding $G_1$ of an abstract graph~$G$ on~$S$ corresponds a
\emphdef{combinatorial map}, storing the combinatorial information of that
embedding.  For graphs cellularly embedded, this notion is well-known, and
there are several essentially equivalent data structures available to store
such combinatorial maps on surfaces~\cite{k-ugpdd-99}, like the gem
representation~\cite{e-dgteg-03,l-gem-82} and the halfedge data structure
of \textsc{Cgal}\footnote{\protect\url{www.cgal.org}}.  All these data
structures have linear complexity in the number of edges of the graph.  The
gem representation stores four \emph{flags} per edge; intuitively, if the
edge is oriented, two flags close to its head, one to its left and one to
its right, and similarly for the tail.  More formally, a flag represents an
incidence between a vertex, an edge, and a face of the embedding.  Three
involutions allow to move from flag to an ``incident'' flag in the graph:
The first one, $i_v$, keeps the same edge-face incidence and moves to the
opposite vertex; the second one, $i_e$, keeps the same vertex-face
incidence and moves to the opposite edge; the last one, $i_f$, keeps the
same vertex-edge incidence and moves to the opposite face.  Also, each flag
has a pointer to the underlying vertex, edge, and face of~$G$.

Two combinatorial maps of the same graph~$G$ are \emphdef{isomorphic} if
there is a bijection~$\varphi$ between their sets of flags $F_1$ and~$F_2$,
commuting with the three involutions, and such that the underlying vertex
(resp., edge) of a flag~$f$ in~$F_1$ is the same as that of~$\varphi(f)$;
such a bijection is called a \emphdef{map isomorphism}.  It is not hard to
test whether two combinatorial maps representing two cellular embeddings
$G_1$ and~$G_2$ of the same abstract graph~$G$ on~$S$ are isomorphic in
linear time.  Furthermore, we have the following lemma.

\begin{lemma}\label{L:combmap}
  The combinatorial maps representing two cellular embeddings $G_1$
  and~$G_2$ of the same abstract graph~$G$ on~$S$ are isomorphic if and
  only if there exists a homeomorphism of~$S$ mapping $G_1$ to $G_2$.
\end{lemma}
\begin{proof}
  If there exists a homeomorphism of~$S$ mapping $G_1$ to~$G_2$, then
  obviously the combinatorial maps of~$G_1$ and~$G_2$ are isomorphic.
  Conversely, any isomorphism of combinatorial maps extends naturally to a
  homeomorphism between the tubular neighborhoods of the graphs.
  (Informally, one can build a disk for each vertex and a strip for each
  edge of~$G$, and attach these disks and strips as prescribed by the
  combinatorial map of~$G_1$, so that their union forms a tubular
  neighborhood of~$G_1$.  And one can similarly do the same for~$G_2$, the
  gluings being combinatorially the same since the combinatorial maps are
  isomorphic.)  This homeomorphism between the two tubular neighborhoods
  extends by radial extension in every disk, and since all the faces are
  disks, we obtain therefore a homeomorphism of the whole surface~$S$.
\end{proof}

We adapt the gem representation (Figure~\ref{fig:Combmap}) to handle graphs
that are not cellularly embedded, by adding the following information.
(For simplicity of exposition, we only consider graphs without isolated
vertices; it is not hard to extend the data structure to handle this case.)
First, each face may be incident to several cycles of~$G$, and stores a
list containing one flag of each such cycle.  Conversely, each flag has a
pointer to the face it belongs to.  Also, a face may have non-zero genus
and contain some boundary components of~$S$, so we store this genus and
number of boundary components of~$S$.

\begin{figure}[htb]
  \centering
  \def\svgwidth{12cm}
 \begingroup%
  \makeatletter%
  \providecommand\color[2][]{%
    \errmessage{(Inkscape) Color is used for the text in Inkscape, but the package 'color.sty' is not loaded}%
    \renewcommand\color[2][]{}%
  }%
  \providecommand\transparent[1]{%
    \errmessage{(Inkscape) Transparency is used (non-zero) for the text in Inkscape, but the package 'transparent.sty' is not loaded}%
    \renewcommand\transparent[1]{}%
  }%
  \providecommand\rotatebox[2]{#2}%
  \ifx\svgwidth\undefined%
    \setlength{\unitlength}{1361.325bp}%
    \ifx\svgscale\undefined%
      \relax%
    \else%
      \setlength{\unitlength}{\unitlength * \real{\svgscale}}%
    \fi%
  \else%
    \setlength{\unitlength}{\svgwidth}%
  \fi%
  \global\let\svgwidth\undefined%
  \global\let\svgscale\undefined%
  \makeatother%
  \begin{picture}(1,0.39265054)%
    \put(0,0){\includegraphics[width=\unitlength]{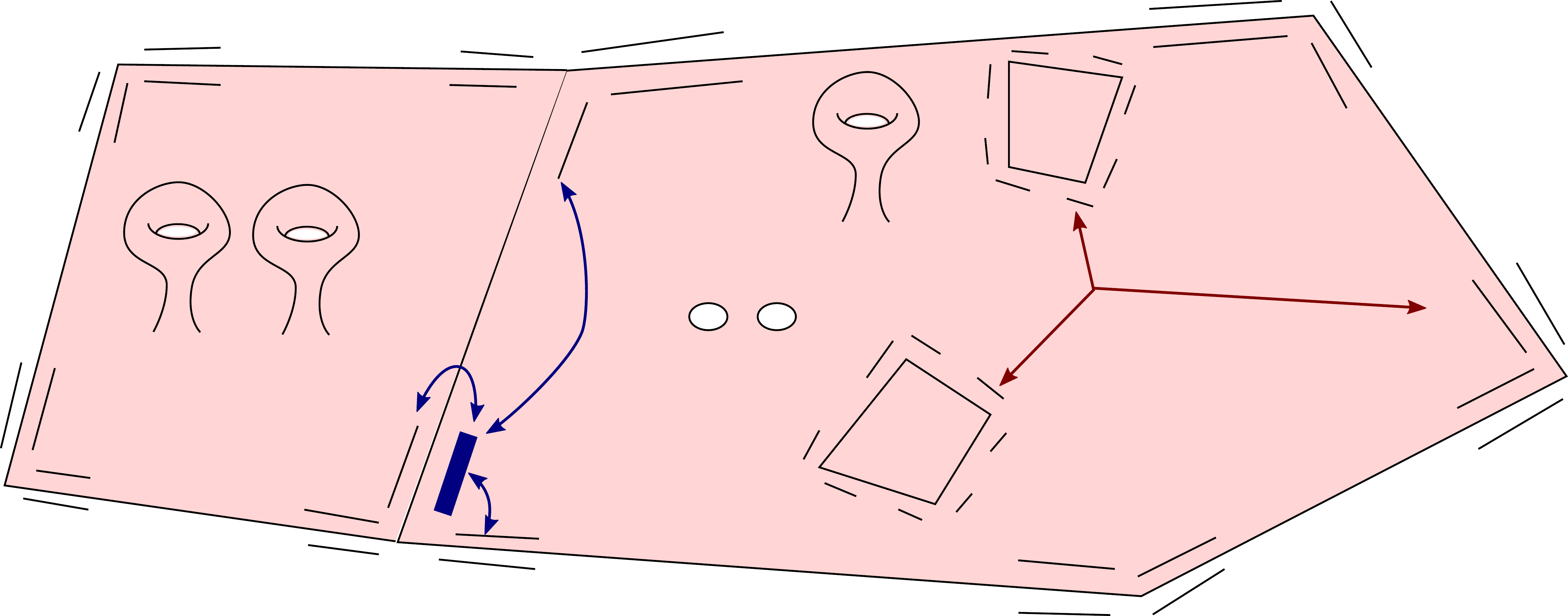}}%
    \put(0.43161468,0.24852066){\color[rgb]{0,0,0}\makebox(0,0)[lb]{\smash{$g=1$}}}%
    \put(0.42993564,0.21493993){\color[rgb]{0,0,0}\makebox(0,0)[lb]{\smash{$b=2$}}}%
    \put(0.11221509,0.1404802){\color[rgb]{0,0,0}\makebox(0,0)[lb]{\smash{$g=2$}}}%
    \put(0.11221509,0.10689946){\color[rgb]{0,0,0}\makebox(0,0)[lb]{\smash{$b=0$}}}%
    \put(0.25108756,0.18057689){\color[rgb]{0,0,0}\makebox(0,0)[lb]{\smash{$i_f$}}}%
    \put(0.39237126,0.16158075){\color[rgb]{0,0,0}\makebox(0,0)[lb]{\smash{$i_v$}}}%
    \put(0.33419562,0.07134912){\color[rgb]{0,0,0}\makebox(0,0)[lb]{\smash{$i_e$}}}%
    \put(0.72599076,0.23875254){\color[rgb]{0,0,0}\makebox(0,0)[lb]{\smash{$p_1$}}}%
    \put(0.79722626,0.17820235){\color[rgb]{0,0,0}\makebox(0,0)[lb]{\smash{$p_2$}}}%
    \put(0.67850045,0.15564444){\color[rgb]{0,0,0}\makebox(0,0)[lb]{\smash{$p_3$}}}%
  \end{picture}%
\endgroup%
  \caption{An example of extended combinatorial map, with the three
    involutions $i_v, i_e$, and $i_f$ for a given flag (in bold), and the
    pointers $p_1, p_2$ and $p_3$ between a face and one flag for each
    incident cycle}
  \label{fig:Combmap}
\end{figure}
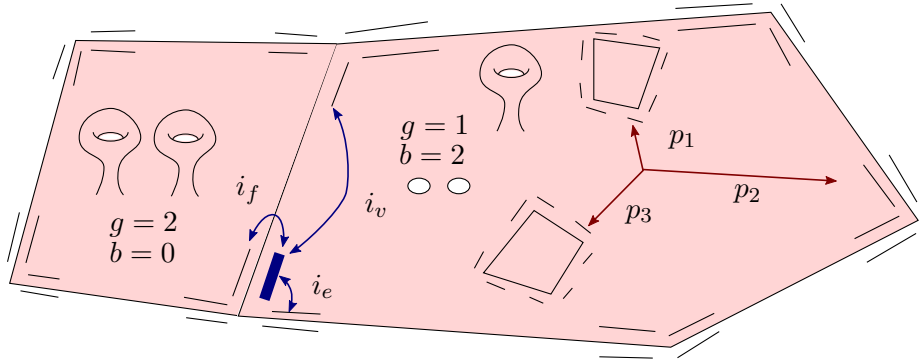

Two such \emphdef{extended combinatorial maps}\footnote{We will drop the
  adjective ``extended'' when it is obvious from the context.} are
\emphdef{isomorphic} if, in addition to the conditions above for standard
combinatorial maps, the corresponding faces have the same genus and the
same number of boundary components of the surface. The corresponding
bijection will be called an \emphdef{extended map isomorphism}. As in the
cellular case, one can check whether two extended combinatorial maps of two
embeddings $G_1$ and~$G_2$ of~$G$ are isomorphic in linear time in their
complexity. And as in the standard case, we have the following lemma.

\begin{lemma}\label{L:extcombtest}
  Two extended combinatorial maps of $G_1$ and $G_2$ are isomorphic if and
  only if there exists a homeomorphism of the surface mapping $G_1$
  to~$G_2$.
\end{lemma}
\begin{proof}
  The proof is similar to that of Lemma~\ref{L:combmap}. The isomorphism
  extends to a homeomorphism between the neighborhoods of the graphs $G_1$
  and $G_2$; since each face of $G_1$ has the same topology as the
  corresponding one in $G_2$, this homeomorphism extends naturally to the
  whole surface.
\end{proof}

Our construction does not depend on whether $S$ is orientable.  If $S$ is
oriented, a flag has a natural \emphdef{orientation}, depending on whether
we turn clockwise or counterclockwise around the vertex of the flag when
starting on the edge of the flag, close to the vertex, and moving towards
the face of the flag.  Furthermore, each of the three involutions reverses
the orientation.  An \emph{orientation} of an extended combinatorial map
assigns an orientation ``clockwise'' or ``counterclockwise'' to each flag
such that each involution reverses the orientation of the flag.  There
exists an \emphdef{orientation-preserving isomorphism} between two oriented
extended combinatorial maps if they are isomorphic as extended
combinatorial maps and the bijection between the flags preserves the
orientation.  This can also be tested in linear time.  Equivalently, there
exists an oriented homeomorphism of the surface mapping one graph embedding
to the other.

%%%%%%%%%%%%%%%%%%%%%%%%%%%%%%%%%%%%%%%%%%%%%%%%%%%%%%%%%%%%%%%%%%%%%%%%%%
\section{Isotopies of Stable Families of Cycles}\label{S:stable}

Let $\Gamma$ be a family of cycles in general position on~$S$.  A
\emphdef{$\bm k$-gon} in~$\Gamma$, for $k\ge1$, is an open disk on~$S$
whose boundary is formed by exactly $k$~subpaths of~$\Gamma$.  A
\emphdef{$\bm 0$-gon} is a disk whose boundary is a single simple cycle
in~$\Gamma$.  In general, $k$-gons may contain and may be crossed by other
pieces of the arrangement of~$\Gamma$; if this is not the case, we say that
the $k$-gon is \emphdef{empty}.

We say that $\Gamma$ is \emphdef{stable} if its arrangement contains no
\emph{empty} $k$-gon for $k\le3$.  In this section, we prove the following
result.
\begin{theorem}\label{T:stable}
  Let $S$ be an orientable surface and let $\Gamma_1=(\gamma_{1,1}, \ldots,
  \gamma_{1,n})$ and $\Gamma_2=(\gamma_{2,1}, \ldots \gamma_{2,n})$ be two
  \emph{stable} families of cycles on~$S$ in general position such that:
  \begin{enumerate}
  \item there exists an oriented homeomorphism~$h$ of $S$ mapping each
    cycle~$\gamma_{1,j}$ of~$\Gamma_1$ to the corresponding
    cycle~$\gamma_{2,j}$ of~$\Gamma_2$ not necessarily pointwise, but
    preserving the orientations of the cycles, and
  \item each cycle of~$\Gamma_1$ is homotopic to the corresponding cycle
    of~$\Gamma_2$.
  \end{enumerate}
  Then there is an isotopy of $S$ that maps each cycle of~$\Gamma_1$ to the
  corresponding cycle of~$\Gamma_2$, not necessarily pointwise, but
  preserving the orientations of the cycles.
\end{theorem}
We note that conversely, if there exists an ambient isotopy of~$S$ mapping
$\Gamma_1$ to~$\Gamma_2$, then conditions (1) and~(2) are satisfied.  

As a side remark, we will actually use this result in a setting where we
know that $h$ maps each cycle of~$\Gamma_1$ \emph{pointwise} to the
corresponding cycle of~$\Gamma_2$; however, even under this stronger
hypothesis, it does not always hold that there exists an isotopy of~$S$
that maps each cycle in~$\Gamma_1$ \emph{pointwise} to the corresponding
cycle of~$\Gamma_2$; see Figure~\ref{F:contreex2}. The following corollary
adds another hypothesis to ensure that this does not happen. It will be not
used directly in this paper, but inspires the techniques used in
Section~\ref{S:ladeg}.

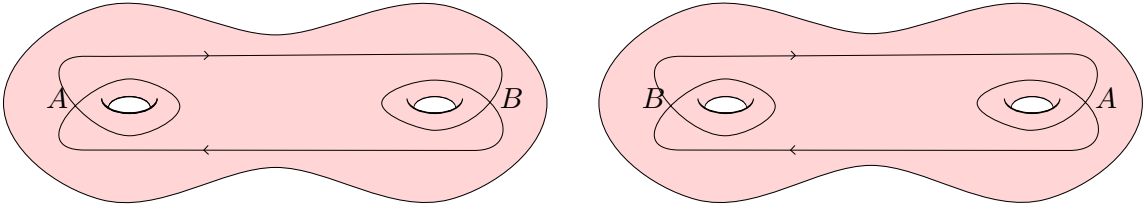
\begin{figure}
\centering
\def\svgwidth{15cm}
\begingroup
  \makeatletter
  \providecommand\color[2][]{%
    \errmessage{(Inkscape) Color is used for the text in Inkscape, but the package 'color.sty' is not loaded}
    \renewcommand\color[2][]{}%
  }
  \providecommand\transparent[1]{%
    \errmessage{(Inkscape) Transparency is used (non-zero) for the text in Inkscape, but the package 'transparent.sty' is not loaded}
    \renewcommand\transparent[1]{}%
  }
  \providecommand\rotatebox[2]{#2}
  \ifx\svgwidth\undefined
    \setlength{\unitlength}{3143.75149967pt}
  \else
    \setlength{\unitlength}{\svgwidth}
  \fi
  \global\let\svgwidth\undefined
  \makeatother
  \begin{picture}(1,0.17730271)%
    \put(0,0){\includegraphics[width=\unitlength]{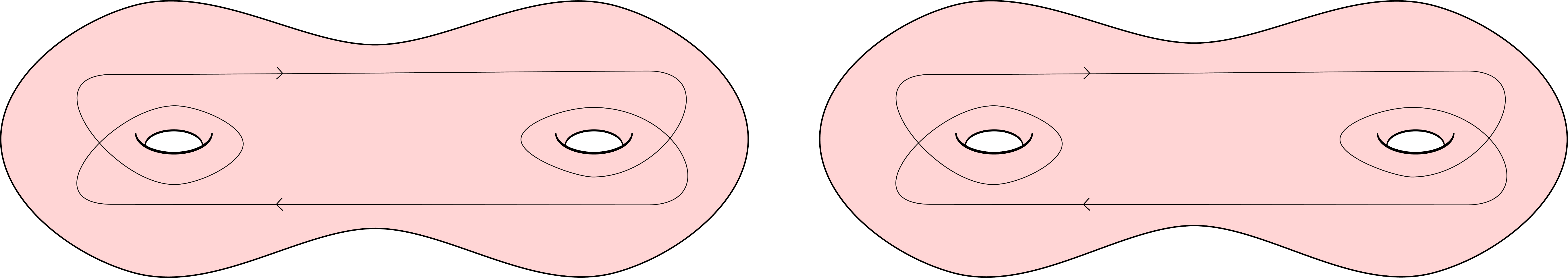}}%
    \put(0.03796587,0.08422273){\color[rgb]{0,0,0}\makebox(0,0)[lb]{\smash{$A$}}}%
    \put(0.43588947,0.08422273){\color[rgb]{0,0,0}\makebox(0,0)[lb]{\smash{$B$}}}%
    \put(0.56030486,0.08422273){\color[rgb]{0,0,0}\makebox(0,0)[lb]{\smash{$B$}}}%
    \put(0.95822844,0.08422273){\color[rgb]{0,0,0}\makebox(0,0)[lb]{\smash{$A$}}}%
  \end{picture}%
\endgroup
\caption{A stable family comprised of a single cycle $\gamma: S^1
  \rightarrow S$, embedded in two different ways on a double torus.
  Letter~$A$ denotes the image of the same point of~$S^1$ in both
  embeddings, and similarly for~$B$.  Both cycles are homotopic, some
  oriented homeomorphism of~$S$ maps one to the other pointwise, and there
  exists an isotopy of~$S$ (namely, the identity) that maps one to the
  other, preserving their orientation; however, there exists no isotopy
  of~$S$ that maps one to the other pointwise}
\label{F:contreex2}
\end{figure}

\begin{corollary}\label{C:stable}
  Let $S$ be an orientable surface and let $\Gamma_1=(\gamma_{1,1}, \ldots,
  \gamma_{1,n})$ and $\Gamma_2=(\gamma_{2,1}, \ldots \gamma_{2,n})$ be two
  \emph{stable} families of cycles on~$S$ in general position such that:
  \begin{enumerate}
  \item there exists an oriented homeomorphism~$h$ of $S$ mapping each
    cycle~$\gamma_{1,j}$ of~$\Gamma_1$ to the corresponding
    cycle~$\gamma_{2,j}$ of~$\Gamma_2$ not necessarily pointwise, but
    preserving the orientations of the cycles, and
  \item each cycle of~$\Gamma_1$ is homotopic to the corresponding cycle
    of~$\Gamma_2$.
  \end{enumerate}
  If there is only one orientation-preserving map automorphism of
  $\Gamma_1$, namely the identity, then there is an ambient isotopy of $S$
  mapping each cycle of~$\Gamma_1$ to the corresponding cycle of~$\Gamma_2$
  pointwise.
\end{corollary}
\begin{proof}
  Theorem~\ref{T:stable} shows that there is an ambient isotopy mapping
  each cycle of~$\Gamma_1$ to the corresponding cycle of~$\Gamma_2$ that is
  not necessarily pointwise, but preserves the orientation of each cycle.
  Hence this ambient isotopy induces an orientation-preserving map
  isomorphism between $\Gamma_1$ and $\Gamma_2$. But the extended
  combinatorial maps of $\Gamma_1$ and $\Gamma_2$ are the same, since there
  exists an oriented homeomorphism between them. We thus get an
  orientation-preserving map automorphism of $\Gamma_1$, and it has to be
  the identity.  This means that each vertex of the arrangement
  of~$\Gamma_1$ is mapped to the corresponding vertex of the arrangement
  of~$\Gamma_2$, and similarly for each edge of these arrangements (which
  are mapped with the same orientation).  After an isotopy in a
  neighborhood of these edges, we get a pointwise ambient isotopy.
\end{proof}

The remaining part of this section is devoted to the proof of Theorem~\ref{T:stable}
if the surface $S$ has negative Euler characteristic. The proof for the
remaining surfaces uses slightly different tools and is deferred to
Section~\ref{S:exc}.
%%%%%%%%%%%%%%%%%%%%%%%%%%%%%%%%%%%%%%%%%%%%%%%%%%%%%%%%%%%%%%%%%%%%%%%%%%
\subsection{Basic Consequences of Euler's Formula}

The following lemmas are simple consequences of Euler's formula.
\begin{lemma}\label{lemeuler2}
  Assume just for this lemma that $S$ is a sphere.  Let $G$ be a connected
  (hence cellularly embedded) graph on $S$, such that every vertex has
  degree four. Then there are at least three faces in $G$ with degree
  smaller than four.
\end{lemma}
\begin{proof}
  Let $v$, $e$, and~$f$ denote the number of vertices, edges, and faces
  of~$G$.  The sum of the degrees of all faces is~$2e$, which, by Euler's
  formula $v-e+f=2$ and double-counting of the vertex-edge incidences
  $4v=2e$, equals $4f-8$.  Since every face has degree at least one, if at
  most two faces have degree smaller than four, the sum of the degrees of
  all faces is at least $4f-6$.  This is a contradiction.
\end{proof}
A trivial corollary is the following:
\begin{corollary}\label{coreuler}
  Assume just for this corollary that $S$ is a sphere, a disk, or an annulus.
  Let $G$ be a connected graph embedded on~$S$, such that every vertex has
  degree four. Then there is at least one face in $G$ with degree smaller
  than four.
\end{corollary}
\begin{proof}
  Let $\bar S$ be the sphere obtained from~$S$ by attaching a disk to each
  of the boundary components of~$S$.  The graph $G$ is embedded on~$\bar
  S$, and by Lemma~\ref{lemeuler2} it contains at least three faces
  (on~$\bar S$) with degree smaller than four.  At least one of these faces
  does not contain the disks attached to~$S$, since there are at most two
  such disks.
\end{proof}

\begin{lemma}\label{lemeuler}
  Let $\Gamma$ be a stable family of cycles in a surface $S$.  Then no
  $k$-gon can exist in~$\Gamma$ for $k\le3$.
\end{lemma}
\begin{proof}
  Consider a hypothetical $k$-gon~$D$ for~$\Gamma$, with $k\le3$.  We prove
  that $D$ strictly contains another $k'$-gon with $k'\le3$, which is a
  contradiction, since by induction this would give an infinite family of
  $k$-gons in a finite graph.

  Assume first that $D$ has exactly $k$ vertices on its boundary; in other
  words, no cycle crosses the boundary of $D$.  Since $D$ cannot be an
  empty $k$-gon, there must be a connected component~$\Gamma'$ of the
  arrangement of~$\Gamma$ that is entirely inside~$D$.
  Corollary~\ref{coreuler} implies that $\Gamma'$ contains empty $k'$-gons
  with $k'\le3$ inside~$D$, and thus $D$ contains a smaller $k'$-gon, as
  desired.

  Assume now that $D$ has at least $k+1$ vertices on its boundary, i.e., it
  intersects at least another cycle in $\Gamma$.  Consider the restriction
  of the arrangement of~$\Gamma$ to the closed disk~$D$, and let $G$ be the
  connected component of that restriction that contains the boundary
  of~$D$.  To prove the lemma, it suffices to prove that $G$ has an
  interior face of degree at most three.  We now assume that every interior
  face of~$G$ has degree at least four, and will reach a contradiction.

  $G$ has two types of edges: $e_{ext}$ \emph{external} edges lying on the
  boundary of the $k$-gon, and $e_{int}$ \emph{internal} edges. Similarly,
  it has three types of vertices: $v_{int}$ vertices in the interior of the
  $k$-gon, $v_{flat}$ degree-three vertices on the boundary of the $k$-gon,
  and $v_{extr}=k$ degree-two vertices on the boundary of the $k$-gon. Let
  $f$ denote the number of interior faces of~$G$ in the plane.  By
  double-counting arguments, we obtain:
  \begin{eqnarray}
    v_{extr}+v_{flat}&=&e_{ext} \label{eqneuler1}\\ 4f
    &\leq&2e_{int}+e_{ext} \label{eqneuler2}\\
    4v_{int}+v_{flat}&=&2e_{int}. \label{eqneuler3}
  \end{eqnarray}
  Euler's formula implies
  \begin{eqnarray*}
    4 &= & 4(v_{extr}+v_{flat}+v_{int})-4(e_{ext}+e_{int})+4f\\
    &\stackrel{(\ref{eqneuler1})}{=}&4v_{int}-4e_{int}+4f\\
    &\stackrel{(\ref{eqneuler3})}{=}&-v_{flat}-2e_{int}+4f\\
    &\stackrel{(\ref{eqneuler2})}{\leq}&-v_{flat}-2e_{int}+2e_{int}+e_{ext}.\\
  \end{eqnarray*}
  With equation $(\ref{eqneuler1})$, this gives $v_{extr} \geq 4$ which
  yields a contradiction and concludes the proof.
\end{proof}

%%%%%%%%%%%%%%%%%%%%%%%%%%%%%%%%%%%%%%%%%%%%%%%%%%%%%%%%%%%%%%%%%%%%%%%%%%
\subsection{Following the Geodesics}\label{S:follow}

We recall that $\Gamma_1=(\gamma_{1,1}, \ldots, \gamma_{1,n})$ and
$\Gamma_2=(\gamma_{2,1}, \ldots \gamma_{2,n})$ are two stable
families of cycles on~$S$ in general position satisfying the hypotheses of
Theorem~\ref{T:stable}. Since we assume that the Euler characteristic
of~$S$ is negative, we can endow~$S$ with a hyperbolic metric and identify
$\widetilde S$ with a subset of the open hyperbolic disk. For each~$j$, if
$\gamma_{1,j}$ is not null-homotopic, let $g_j$ be the unique geodesic
homotopic to it.  (Some $g_j$'s may be identical.)
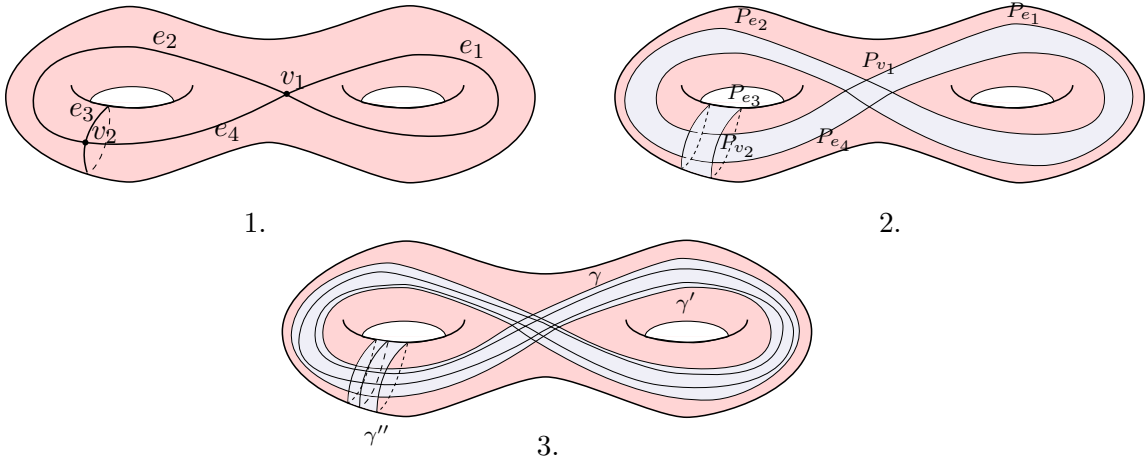
\begin{figure}[htb]
  \centering
  \def\svgwidth{15cm}
\begingroup
  \makeatletter
  \providecommand\color[2][]{%
    \errmessage{(Inkscape) Color is used for the text in Inkscape, but the package 'color.sty' is not loaded}
    \renewcommand\color[2][]{}%
  }
  \providecommand\transparent[1]{%
    \errmessage{(Inkscape) Transparency is used (non-zero) for the text in Inkscape, but the package 'transparent.sty' is not loaded}
    \renewcommand\transparent[1]{}%
  }
  \providecommand\rotatebox[2]{#2}
  \ifx\svgwidth\undefined
    \setlength{\unitlength}{2105.05584pt}
  \else
    \setlength{\unitlength}{\svgwidth}
  \fi
  \global\let\svgwidth\undefined
  \makeatother
  \begin{picture}(1,0.39840485)%
    \put(0,0){\includegraphics[width=\unitlength]{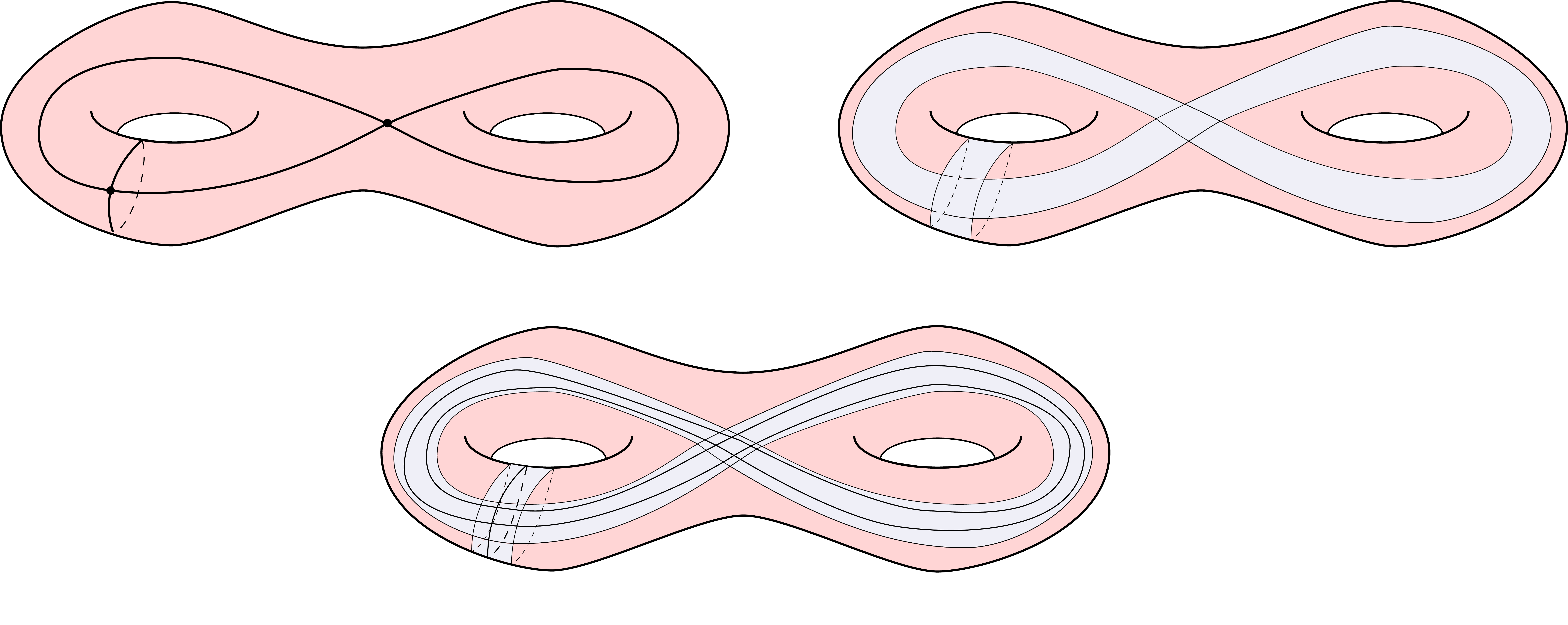}}%
    \put(0.75136359,0.3439438){\color[rgb]{0,0,0}\makebox(0,0)[lb]{\smash{${\scriptstyle {P_{v_1}}}$}}}%
    \put(0.241321,0.32824224){\color[rgb]{0,0,0}\makebox(0,0)[lb]{\smash{$v_1$}}}%
    \put(0.07586501,0.28022193){\color[rgb]{0,0,0}\makebox(0,0)[lb]{\smash{$v_2$}}}%
    \put(0.39833459,0.35584553){\color[rgb]{0,0,0}\makebox(0,0)[lb]{\smash{$e_1$}}}%
    \put(0.1290509,0.36398918){\color[rgb]{0,0,0}\makebox(0,0)[lb]{\smash{$e_2$}}}%
    \put(0.05615785,0.30001888){\color[rgb]{0,0,0}\makebox(0,0)[lb]{\smash{$e_3$}}}%
    \put(0.18361016,0.28222325){\color[rgb]{0,0,0}\makebox(0,0)[lb]{\smash{$e_4$}}}%
    \put(0.62671127,0.27099754){\color[rgb]{0,0,0}\makebox(0,0)[lb]{\smash{${\scriptstyle P_{v_2}}$}}}%
    \put(0.63940344,0.38287484){\color[rgb]{0,0,0}\makebox(0,0)[lb]{\smash{$\scriptstyle {P_{e_2}}$}}}%
    \put(0.71134903,0.27607327){\color[rgb]{0,0,0}\makebox(0,0)[lb]{\smash{${\scriptstyle P_{e_4}}$}}}%
    \put(0.87836492,0.3882301){\color[rgb]{0,0,0}\makebox(0,0)[lb]{\smash{${\scriptstyle P_{e_1}}$}}}%
    \put(0.63341351,0.31535245){\color[rgb]{0,0,0}\makebox(0,0)[lb]{\smash{${\scriptstyle P_{e_3}}$}}}%
    \put(0.2094017,0.19785854){\color[rgb]{0,0,0}\makebox(0,0)[lb]{\smash{$1.$}}}%
    \put(0.76669943,0.19785854){\color[rgb]{0,0,0}\makebox(0,0)[lb]{\smash{$2.$}}}%
    \put(0.46646987,0.00078197){\color[rgb]{0,0,0}\makebox(0,0)[lb]{\smash{$3.$}}}%
    \put(0.5120239,0.15442957){\color[rgb]{0,0,0}\makebox(0,0)[lb]{\smash{$\scriptstyle {\gamma}$}}}%
    \put(0.58880307,0.12986023){\color[rgb]{0,0,0}\makebox(0,0)[lb]{\smash{$\scriptstyle{\gamma'}$}}}%
    \put(0.31508537,0.01277203){\color[rgb]{0,0,0}\makebox(0,0)[lb]{\smash{$\scriptstyle{\gamma''}$}}}%
  \end{picture}%
\endgroup
  \caption{1. Two geodesics on a double torus, and the associated
    graph. 2. The corridors associated to them. 3. Cycles in the
    corridors. In vertex polygons, they intersect in a grid-like pattern}
  \label{fig:Corridors}
\end{figure}

The following result uses a proposition by de Graaf and
Schrijver~\cite[Proposition 13]{gs-mcmcr-97}, the proof of which is based
on a slight refinement of Ringel's theorem~\cite{r-tegtg-55} for
arrangements of curves in the disk. It allows us to push each curve in $\Gamma_1$ and
$\Gamma_2$ in the neighborhood of the corresponding geodesic.
\begin{proposition}\label{proposition13}
  No cycle in~$\Gamma_1$ and~$\Gamma_2$ is null-homotopic. Furthermore, for
  $i=1,2$, for each $\varepsilon>0$, up to replacing $\Gamma_i$ with its
  image by an ambient isotopy of~$S$, we may assume that each
  cycle~$\gamma_{i,j}$ has a lift that belongs to an
  $\varepsilon$-neighborhood of a lift of~$g_j$.
\end{proposition}
\begin{proof}
  It follows from the aforementioned result by de Graaf and
  Schrijver~\cite[Proposition 13]{gs-mcmcr-97} that the cycles
  $\gamma_{i,j}$ in $\Gamma_i$ can be moved through Reidemeister moves (not
  increasing the number of crossings) and ambient isotopies into cycles
  which have lifts that are $\varepsilon$-close to some lift of a geodesic~$\alpha_j$
  if $\gamma_{i,j}$ is not null-homotopic, and $\varepsilon$-close to a
  point of~$\wt{S}$ otherwise. As $\Gamma_i$ is stable, no Reidemeister
  move \emph{at all} is possible, so the only possible moves are actually
  isotopies of the surface.

  Now, assume that $\gamma_{i,j}$ is null-homotopic. For $\varepsilon$
  small enough, $\gamma_{i,j}$ is a contractible cycle in a disk; hence if
  it is simple, it forms a $0$-gon; and if it is non-simple, by
  Corollary~\ref{coreuler}, it contains a $k$-gon for some $k\le3$. This
  contradicts the stability of $\Gamma_1$ and $\Gamma_2$, hence no cycle
  $\gamma_{i,j}$ is null-homotopic.

  Now, since lifts of $\gamma_{i,j}$ and $\alpha_j$ are $\varepsilon$-close,
  they share the same endpoints and thus $\alpha_j$ is the unique geodesic
  homotopic to $\gamma_{i,j}$, which shows that $\alpha_j=g_j$ and
  concludes the proof.
\end{proof}
So henceforth we assume that $\Gamma_1$ and~$\Gamma_2$ satisfy the
conclusion of the above proposition.

The union of the geodesics~$g_j$ forms a graph, possibly with simple cycles
without vertices, on the surface $S$; we denote by $E$ and $V$ its edges and
vertices, see Figure~\ref{fig:Corridors}. (Simple cycles without vertex are
considered to be closed edges.) Each vertex has even degree, and each
geodesic arriving at a vertex from an edge leaves it via the opposite edge.

Now, following ideas by de Graaf and Schrijver~\cite{gs-mcmcr-97}, we
introduce a polygonal decomposition of an $\varepsilon$-neighborhood of the
graph $(V,E)$; see Figure~\ref{fig:Corridors}.  To each edge~$e\in E$, we
associate an \emphdef{edge polygon}~$P_e$ (actually, a quadrilateral), and
to each vertex $v\in V$ of degree~$2d$, we associate a \emphdef{vertex
  polygon} $P_v$ with $2d$ sides, such that each edge~$e=uv$ lies in the
interior of $P_e\cup P_u \cup P_v$ and $v$ lies in the interior of $P_v$,
and such that the union of all the polygons forms a tubular neighborhood of
the cycles $g_i$.%
\footnote{In the case where a geodesic coincides with a boundary of the
  surface, one of the edges of each polygon of that geodesic actually lies
  on this boundary.  Also, for simple cycles without vertex, we introduce
  an edge polygon with two opposite sides glued together.}  %
We can assume that all the polygons are mutually disjoint, except for $P_e$
and~$P_u$ for $u$ a vertex incident to~$e$, which share an edge.

Let us call a \emph{corridor} the polygonal neighborhood of a single
geodesic~$g_i$, that is $\bigcup P_e \cup \bigcup P_u$ for all $e$ and $u$
in~$g_i$.  Every cycle in~$\Gamma_i$ belongs to a single corridor.

The following proposition is proved using Lemmas~\ref{lemeuler2}
and~\ref{lemeuler} and by standard bigon flipping arguments.
\begin{proposition}\label{P:trgle}
  For $i=1,2$, up to replacing $\Gamma_i$ with its image under an ambient
  isotopy of~$S$, we may assume that:
  \begin{itemize}
  \item each maximal piece of a cycle in~$\Gamma_i$ within a polygon
    is simple, and has its endpoints on opposite sides of the
    polygon;
  \item two such pieces cross at most once; moreover, if they cross, then
    the four endpoints of these two pieces are all in different sides of
    the polygon (in particular, the polygon is a vertex polygon).
  \end{itemize}
\end{proposition}
\begin{proof}
  If a maximal piece of~$\Gamma_i$ within a polygon is non-simple, it forms
  a 1-gon, which is impossible (Lemma~\ref{lemeuler}).  If two pieces cross
  twice, they form a 2-gon, which is impossible for the same reason.

  If such a piece has its endpoints on two different sides of a polygon
  that are not opposite in that polygon, that polygon is a vertex polygon,
  and the corresponding cycle does not belong to a single corridor, which
  is impossible by construction.

  If such a piece~$p$ has its endpoints on the same side~$s$ of a polygon,
  without loss of generality assume that the disk~$D$ bounded by~$s$
  and~$p$ contains no other piece with both endpoints on~$s$.  Thus the
  pieces inside~$D$ are simple, pairwise disjoint (otherwise these two
  pieces would form a 3-gon with~$p$), and connect $s$ to~$p$; so with an
  isotopy of the surface, we can push~$p$ across~$s$, decreasing the total
  number of intersections between the cycles and the sides of the polygon.
  After finitely many such operations, no such piece~$p$ exists.

  There only remains to prove that there cannot be any (self-)intersection
  among cycles in the same corridor~$C$.  We distinguish two cases:
  \begin{itemize}
  \item If $C$ contains at least one vertex polygon, some cycle crosses
    every cycle in~$C$; if there were a \mbox{(self-)}intersection in~$C$, there
    would be a 3-gon.
  \item On the other hand, if the corridor $C$ contains no vertex polygon,
    then $C$ is an annulus.  Consider the arrangement of the cycles
    in~$\Gamma$, and assume there is a crossing.  One connected component
    of this arrangement is a graph where all vertices have degree four.
    Thus, by Corollary~\ref{coreuler}, $\Gamma$ contains a $k$-gon for
    $k\le3$, contradicting Lemma~\ref{lemeuler}.\qedhere
  \end{itemize}
\end{proof}

We can now assume that $\Gamma_1$ and~$\Gamma_2$ satisfy the conclusion of
Proposition~\ref{P:trgle}.  It follows that all the arcs of~$\Gamma_i$ in a
given edge polygon~$P_e$ are simple, disjoint, and belong to different
cycles.  Moreover, each polygon~$P_v$ is actually a quadrilateral, because
otherwise three arcs in~$P_v$ coming from six different sides would cross
inside~$P_v$, yielding a 3-gon because of the general position assumption,
which is impossible (Lemma~\ref{lemeuler}).  Finally, within each
polygon~$P_v$, the arcs intersect in a grid-like fashion, as in
Figure~\ref{fig:Corridors}.

%%%%%%%%%%%%%%%%%%%%%%%%%%%%%%%%%%%%%%%%%%%%%%%%%%%%%%%%%%%%%%%%%%%%%%%%%%
\subsection{A Technical Result on Corridors}\label{S:tech}

Henceforth, let us choose an arbitrary orientation on~$S$.  Let $C$ be a
corridor, oriented in the direction of its geodesic~$g$; let $C_l$
and~$C_r$ be the left and right boundaries of~$C$, respectively. We recall
that $h$ is the oriented homeomorphism specified in the hypotheses of Theorem~\ref{T:stable}.

Let $\Gamma_1^C$ be the subfamily of cycles in~$\Gamma_1$ that belong
to~$C$, and let $\Gamma_2^C$ be its image by~$h$.  Recall that $\Gamma_i^C$
has no crossing in an edge polygon~$P_e$.  For $i=1,2$, the \emph{ordering}
of~$\Gamma_i^C$ along~$C$ is defined as follows: Consider an arc in an edge
polygon~$P_e$, going from~$C_l$ to~$C_r$, crossing each cycle
in~$\Gamma_i^C$ exactly once, and record the index of the cycles
in~$\Gamma_i$ encountered, in this order along the arc.  By construction,
this ordering does not depend on the choice of the polygon and arc.

\begin{lemma}\label{L:homeo}
  The orderings of $\Gamma_1^C$ and~$\Gamma_2^C$ are the same.
\end{lemma}
\begin{proof}
  We first claim that the oriented homeomorphism $h: S \rightarrow S$ lifts
  to an oriented homeomorphism $\widetilde{h}:\wS\rightarrow\wS$.  Indeed,
  if we denote by~$\pi$ the projection $\pi: \wS \rightarrow S$ and apply
  the lifting theorem~\cite[Proposition 1.33]{h-at-02} to $h\circ\pi$, we
  get%
  \footnote{The hypotheses are trivially fulfilled since the universal
    cover has trivial fundamental group.} %
  a continuous map $\widetilde{h}:\wS\rightarrow\wS$ satisfying $\pi\circ
  \wt{h}=h\circ\pi$.  Let $x\in\wS$ and $y=\wt h(x)$.  Similarly, we lift
  $h^{-1}$ to a continuous map~$\wt{h^{-1}}$ such that $\wt{h^{-1}}(y)=x$.
  Then $\wt{h^{-1}}$ and $\wt{h}$ are inverse continuous maps on $\wS$, so
  $\wt{h}$ is a homeomorphism.  Furthermore, $\wt h$ is oriented because
  $h$ is oriented.

  Let $\wt g$ be the lift of a geodesic inside a lift~$\wt C$ of~$C$.  The
  homeomorphism~$\wt h$ maps~$\wt g$ into a possibly different lift of~$C$.
  However, up to composing $\wt h$ with a deck transformation of~$\wS$, we
  may assume that $\wt h$ maps $\wt g$ into~$\wt C$.  Furthermore, $\wt g$
  and its image by~$\wt h$ have the same orientation in~$\wt C$, because
  otherwise the endpoints of each lift of~$\Gamma_1^C$ in~$\wt C$ would be
  exchanged under~$\wt h$, which is not the case since $h$ preserves the
  homotopy classes of the cycles in~$\Gamma_1^C$.  Therefore, $\wt h$ maps
  $\wt g$ to a path in~$\wC$ with the same source and target.  Since $\wt
  h$ is oriented, the orderings of the cycles in $\wt\Gamma_1^C$
  and~$\wt\Gamma_2^C$, from left to right in~$\wt C$, are the same.  It
  follows that they are also the same in~$C$.
\end{proof}

%%%%%%%%%%%%%%%%%%%%%%%%%%%%%%%%%%%%%%%%%%%%%%%%%%%%%%%%%%%%%%%%%%%%%%%%%%
\subsection{End of Proof}

\begin{proof}[Proof of Theorem~\ref{T:stable}]
  Let $P_v$ and~$P_e$ be incident vertex and edge polygons, respectively,
  and let $p$ be the path that is their common boundary.  We first build an
  isotopy of the surface such that, when restricting to~$p$, the image of
  each cycle in~$\Gamma_1$ is the same as the corresponding cycle
  in~$\Gamma_2$.  For this purpose, note that the restriction of~$\Gamma_1$
  to~$p$ is a finite set of points, and similarly for~$\Gamma_2$;
  furthermore, the numbers of points are the same (by Lemma~\ref{L:homeo}).
  We can easily push the intersection points on~$p$ so that they coincide,
  by an isotopy of~$S$ that is the identity outside a neighborhood of~$p$.
  Lemma~\ref{L:homeo} now implies that, after this isotopy, each arc
  in~$P_e$ corresponds to the same cycle in~$\Gamma_1$ and~$\Gamma_2$.

  We can do this operation for every intersection $p$ of a vertex and an
  edge polygon.  Now, within each edge polygon, the arcs of~$\Gamma_1$ are
  simple, pairwise disjoint and in the same order as the arcs
  of~$\Gamma_2$; thus, there exists a homeomorphism from~$P_e$ to~$P_e$
  that is the identity on its boundary and maps the image of~$\Gamma_1$
  inside~$P_e$ to the image of~$\Gamma_2$ inside~$P_e$.  By Alexander's
  lemma, this homeomorphism is an ambient isotopy.  Now, within each edge
  polygon~$P_e$, the images of~$\Gamma_1$ and~$\Gamma_2$ are the same, and
  each arc corresponds to the same cycle in~$\Gamma_1$ and~$\Gamma_2$.

  Now, within each vertex polygon~$P_v$, the endpoints of the arcs
  of~$\Gamma_1$ and~$\Gamma_2$ coincide; moreover, they form
  combinatorially isomorphic arrangements of arcs (namely, grids)
  inside~$P_v$.  The same argument as above shows that an isotopy of~$P_v$
  maps the arcs of~$\Gamma_1$ to the arcs of~$\Gamma_2$.

  Finally, we have found an ambient isotopy~$i$ of~$S$ that maps each cycle
  $\gamma_{1,j}$ in~$\Gamma_1$ to the corresponding cycle~$\gamma_{2,j}$
  in~$\Gamma_2$, as sets but not necessarily pointwise.  Furthermore, since
  $i(\gamma_{1,j})$ is homotopic to $\gamma_{1,j}$, it is also homotopic to
  $\gamma_{2,j}$, so the ambient isotopy~$i$ preserves the orientations of
  the cycles.
\end{proof}

%%%%%%%%%%%%%%%%%%%%%%%%%%%%%%%%%%%%%%%%%%%%%%%%%%%%%%%%%%%%%%%%%%%%%%%%%%
\section{Isotopies of Graph Embeddings}\label{S:ladeg}

In this section, we prove the following result.
\begin{theorem}\label{T:ladeg}
  Let $G_1$ and~$G_2$ be two graph embeddings of a graph~$G$ on an
  orientable surface~$S$.  Assume that there is an oriented
  homeomorphism~$h$ of~$S$ mapping $G_1$ to~$G_2$.  There exists a
  family~$\Lambda$ of cycles in~$G$ such that the following holds: If, for
  each cycle $\gamma$ in~$\Lambda$, the images of~$\gamma$ in $G_1$
  and~$G_2$ are homotopic, then there exists an ambient isotopy of~$S$
  taking $G_1$ to~$G_2$ pointwise.

  Furthermore, the cycles in~$\Lambda$ use each edge of~$G$ at most four times in
  total and, given only the combinatorial map of~$G_1$ on~$S$, one can
  compute the cycles of~$\Lambda$ in linear time in the complexity of that
  combinatorial map.
\end{theorem}
Note that, by the homeomorphism condition, the combinatorial maps of~$G_1$
and~$G_2$ on~$S$ have to be the same. We also emphasize that in contrast to
Theorem~\ref{T:stable}, the ambient isotopy we obtain is pointwise.

Conversely, if $G_1$ and~$G_2$ are isotopic (in particular, if there is an
ambient isotopy taking $G_1$ to~$G_2$), there must exist an oriented
homeomorphism mapping one to the other, and the images of any cycle of~$G$
in $G_1$ and~$G_2$ are homotopic.  Therefore, Theorem~\ref{T:ladeg} implies
Ladegaillerie's result~\cite{l-cip1c-84} stated in the introduction, and also:
\begin{corollary}\label{C:ext-cont}
  Let $G_1$ and~$G_2$ be two graph embeddings of a graph~$G$ in the
  interior of an orientable surface~$S$.  Assume that there exists an
  isotopy between $G_1$ and~$G_2$.  Then there exists an ambient isotopy
  of~$S$ between $G_1$ to~$G_2$.
\end{corollary}

In our proof of Theorem~\ref{T:ladeg}, if the input graph embeddings are
piecewise-linear with respect to a fixed triangulation of~$S$ (which we can
assume, after an ambient isotopy, by using techniques as in
Epstein~\cite[Appendix]{e-c2mi-66}), our ambient isotopy can be chosen so
as to be piecewise-linear.  In particular:
\begin{corollary}\label{C:ext-pl}
  Let $G_1$ and~$G_2$ be two piecewise-linear graph embeddings of a
  graph~$G$ in the interior of an orientable surface~$S$.  Assume that
  there exists a (not necessarily piecewise-linear) isotopy between $G_1$
  and~$G_2$.  Then there exists a piecewise-linear ambient isotopy of~$S$
  between $G_1$ to~$G_2$.
\end{corollary}

For the proof of Theorem~\ref{T:ladeg}, the difficulty of the construction
resides in the fact that the families $\Lambda_1$ and~$\Lambda_2$ (the
images of $\Lambda$ in $G_1$ and $G_2$) must have small complexity.  In a
sense, $\Lambda_1$ forms a topological decomposition of the tubular
neighborhood of~$G_1$ using cycles.  However, all known topological
decompositions of surfaces made of cycles (like pants
decompositions~\cite{cl-oslos-05}, octagonal
decompositions~\cite{ce-tnpcs-10}, or systems of loops~\cite{ew-gohhg-05})
have worst-case complexity $\Omega(gn)$, where $n$ is the complexity of the
surface; our construction has linear size in the complexity of the object
studied.  We suspect that our construction can be useful for other purposes
as well.

%%%%%%%%%%%%%%%%%%%%%%%%%%%%%%%%%%%%%%%%%%%%%%%%%%%%%%%%%%%%%%%%%%%%%%%%%%
\subsection{Preprocessing Step}\label{S:preproc}

For the proof of Theorem~\ref{T:ladeg}, we assume for simplicity of
exposition that $G_1$ and~$G_2$ are known.  It is immediate to check that,
actually, only the combinatorial map of $G_1$ on~$S$ is needed in the
constructions.

\begin{proposition}\label{P:preproc}
  Without loss of generality, we may assume that (1) $G_1$ has a single
  face, or none of its faces is a disk, and (2) $G_1$ has no vertex of
  degree zero or one. Via the oriented homeomorphism $h$, the same holds
  for $G_2$.
\end{proposition}

Intuitively, the proof is simple: Whenever $e$ is an edge of~$G$ bounding
two different faces, at least one of which is a disk, removing $e$ in $G_1$
and~$G_2$ does not change whether $G_1$ and~$G_2$ are isotopic.  Similarly,
removing vertices of degree zero has no effect on the existence of an isotopy.

In more detail, we will need the following two lemmas.
\begin{lemma}\label{L:preproc-disk}
  Let $e$ be an edge of~$G$ bounding two different faces, at least one of
  which is a disk, in the embedding~$G_1$ (and thus also in~$G_2$).  Let
  $G'_1$ and $G'_2$ be the embedded graphs obtained after the removal
  of~$e$.  Then $G_1$ and~$G_2$ are ambient isotopic if and only if $G'_1$
  and~$G'_2$ are ambient isotopic.
\end{lemma}
\begin{proof}
  The direct implication is obvious. Now, assume we have an isotopy~$i$
  mapping $G'_1$ to $G'_2$; we want to deduce that there is an isotopy
  mapping $e_1$ to~$e_2$ (the images of $e$ in $G_1$ and $G_2$).  By
  composition with~$i$, we may assume that $G'_1=G'_2$.  By the existence
  of~$h$, we know that $e_1$ and~$e_2$ are arcs with the same endpoints in
  the same face of $G'_1=G'_2$; furthermore, that face is split into two
  pieces, one of which is a disk, by $e_1$ (resp., $e_2$).

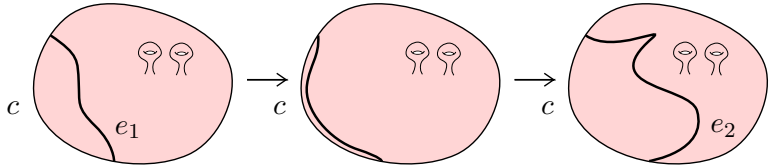
\begin{figure}[htb]
  \centering
  \def\svgwidth{10cm}
\begingroup%
  \makeatletter%
  \providecommand\color[2][]{%
    \errmessage{(Inkscape) Color is used for the text in Inkscape, but the package 'color.sty' is not loaded}%
    \renewcommand\color[2][]{}%
  }%
  \providecommand\transparent[1]{%
    \errmessage{(Inkscape) Transparency is used (non-zero) for the text in Inkscape, but the package 'transparent.sty' is not loaded}%
    \renewcommand\transparent[1]{}%
  }%
  \providecommand\rotatebox[2]{#2}%
  \ifx\svgwidth\undefined%
    \setlength{\unitlength}{437.9671959bp}%
    \ifx\svgscale\undefined%
      \relax%
    \else%
      \setlength{\unitlength}{\unitlength * \real{\svgscale}}%
    \fi%
  \else%
    \setlength{\unitlength}{\svgwidth}%
  \fi%
  \global\let\svgwidth\undefined%
  \global\let\svgscale\undefined%
  \makeatother%
  \begin{picture}(1,0.21411942)%
    \put(0,0){\includegraphics[width=\unitlength]{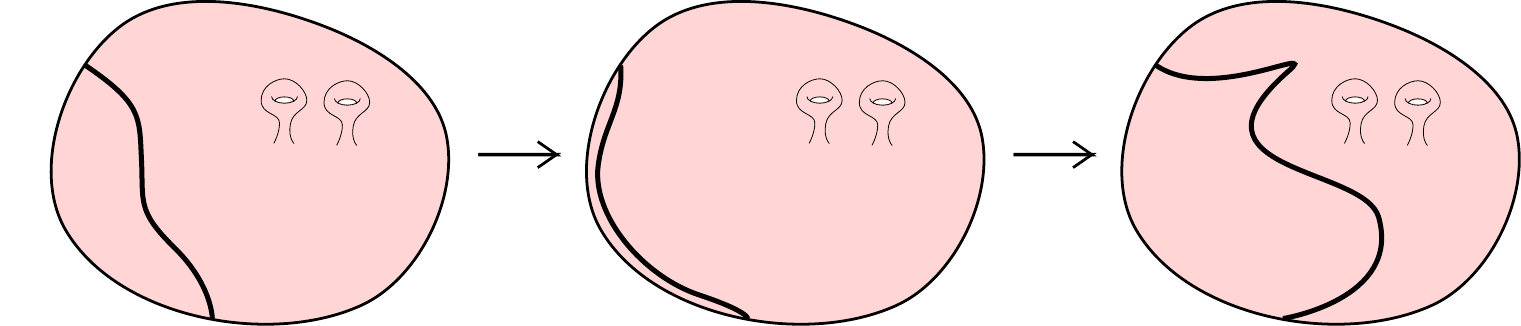}}%
    \put(0.14068126,0.04089621){\color[rgb]{0,0,0}\makebox(0,0)[lb]{\smash{$e_1$}}}%
    \put(0.92443534,0.04272314){\color[rgb]{0,0,0}\makebox(0,0)[lb]{\smash{$e_2$}}}%
    \put(-0.00181949,0.06647325){\color[rgb]{0,0,0}\makebox(0,0)[lb]{\smash{$c$}}}%
    \put(0.34895156,0.06647325){\color[rgb]{0,0,0}\makebox(0,0)[lb]{\smash{$c$}}}%
    \put(0.70154953,0.06647325){\color[rgb]{0,0,0}\makebox(0,0)[lb]{\smash{$c$}}}%
  \end{picture}%
\endgroup%
  \caption{Illustration of the proof of Lemma~\ref{L:preproc-disk}.}
  \label{fig:Preproc}
\end{figure}
The rest of the proof is illustrated in Figure~\ref{fig:Preproc}.  Since
one of the faces bounded by $e_1$ is a disk, $e_1$ can be isotoped (with
fixed extremities) to a neighborhood of the curve~$c$ closing this disk,
and the same goes for $e_2$.  After this isotopy, consider a disk
neighborhood of~$c$ containing both $e_1$ and~$e_2$.  Since these two edges
have the same endpoints in this disk, the Jordan--Sch\"onflies theorem
implies that there is a homeomorphism of the disk, fixed on the boundary,
that maps one to another; then, by Alexander's lemma, that homeomorphism
can be obtained by an isotopy of the disk.
\end{proof}
\begin{lemma}\label{L:preproc-tree}
  Let $v$ be a vertex of~$G$ of degree one, and let $e$ be its incident
  edge.  Let $G'_1$ and $G'_2$ be the embedded graphs obtained after the
  removal of $e$ and~$v$.  Then $G_1$ and~$G_2$ are ambient isotopic if and
  only if $G'_1$ and~$G'_2$ are ambient isotopic.
\end{lemma}
\begin{proof}
  Again, one direction is trivial; the converse can be proved using similar
  ideas as the previous lemma.  Here, the topological statement that is
  used is the following: Let~$p$ be a point on the boundary of a disk, and
  let $e_1$ and~$e_2$ be two simple paths having $p$~as an endpoint and
  intersecting the boundary of the disk exactly at~$p$; then there is an
  ambient isotopy of the disk, fixed on its boundary, that maps $e_1$
  to~$e_2$.  This again follows by an application of the
  Jordan--Sch\"onflies theorem and Alexander's lemma (by first extending
  $e_1$ and~$e_2$ to simple arcs with the same endpoints).
\end{proof}
\begin{proof}[Proof of Proposition~\ref{P:preproc}]
  We show below how to build in linear time a subgraph $G''$ of~$G$
  satisfying the desired properties and such that, if $G''_1$ (resp.,
  $G''_2$) denotes the restriction of~$G_1$ (resp., $G_2$) to~$G''$, then
  $G_1$ and~$G_2$ are ambient isotopic if and only if $G''_1$ and~$G''_2$
  are ambient isotopic.  This is enough to prove the proposition.

  Let $G=(V,E)$.  We initially set $E':=E$, and, for each edge of~$E'$ in
  turn, we remove it from~$E'$ if and only if it is incident to two
  distinct faces of $(V,E')$, at least one of which is a disk (in the
  embedding $G_1$ or~$G_2$).  This is easy to do in linear time, by
  initially labeling each face of $(V,E')$ with its topology (genus and
  number of boundary components) and maintaining this labeling during the
  process.

  Let $G'_1$ and~$G'_2$ be the embeddings of~$(V,E')$ induced by $G_1$
  and~$G_2$, respectively.  Lemma~\ref{L:preproc-disk} implies that $G_1$
  and~$G_2$ are isotopic if and only if $G'_1$ and~$G'_2$ are isotopic.
  Furthermore, if $G'_1$ has at least two faces, one of which is a disk,
  there exists an edge in~$G'_1$ incident to a disk and to another face;
  such an edge would have been removed in the process, which is a
  contradiction.  So the first condition is satisfied.

  Moreover, we can, in linear time, iteratively remove all degree-one
  vertices with their incident edges, until no degree-one vertex remains.
  (Put all degree-one vertices in any list-type data structure; while the
  structure is non-empty, extract any vertex; if it still has degree one,
  remove it with its incident edge; if the opposite vertex on that edge has
  now degree one, add it to the structure; repeat.)
  Lemma~\ref{L:preproc-tree} implies that $G_1$ and~$G_2$ are isotopic if
  and only if these new graph embeddings, $G''_1$ and~$G''_2$, are
  isotopic.

  Finally, if $G''_1$ (and $G''_2$) have isolated vertices, we can safely
  remove them: Since there is a homeomorphism of~$S$ taking $G_1$ to~$G_2$,
  the isolated vertices belong to the same faces in both embeddings.
\end{proof}

Proposition~\ref{P:preproc} leads us to distinguish two cases, leading to
slightly different constructions depending on whether after the
preprocessing, $G_1$ has a single face, or none of its faces are disks.  We
describe the latter first, as the former will build on it; it will be
described in Section~\ref{S:cutgraph}.

%%%%%%%%%%%%%%%%%%%%%%%%%%%%%%%%%%%%%%%%%%%%%%%%%%%%%%%%%%%%%%%%%%%%%%%%%%
\subsection{Proof of Theorem~\ref{T:ladeg} if no face is a disk}\label{S:nodisk}

In this section, we prove Theorem~\ref{T:ladeg} in the special case where
no face of~$G_1$ (or, equivalently, $G_2$) is a disk.  We can assume
without loss of generality that $G_1$ satisfies the properties of
Proposition~\ref{P:preproc}.

\paragraph*{Construction of the Stable
  Family~$\Gamma$.}\label{S:construction}

\ We first build a family~$\Gamma$ of cycles in~$G$ whose images in $G_1$
or~$G_2$ are slight perturbations of stable families $\Gamma_1$
and~$\Gamma_2$.  If the images of each cycle in~$\Gamma$ in $G_1$ and~$G_2$
are homotopic, then this almost implies that $G_1$ and~$G_2$ are isotopic,
which suffices to prove Theorem~\ref{T:ladeg}.  Unfortunately, this is not
entirely true, and we need to test a larger family $\Lambda\supset\Gamma$
for homotopy.
\begin{proposition}\label{P:constr}
  In linear time, we can construct a family of cycles~$\Gamma$ in~$G$ such
  that:
  \begin{itemize}
  \item each edge of~$G$ is used at most twice by all the cycles
    in~$\Gamma$;
  \item there exists a \emph{stable} family~$\Gamma_1$ on~$S$ whose cycles
    are homotopic (by an arbitrarily small perturbation) to the cycles in
    the images of~$\Gamma$ in~$G_1$;
  \item $G_1$ does not meet the interior of the faces of the arrangement
    of~$\Gamma_1$ that are not disks.
  \end{itemize}
\end{proposition}
\begin{proof}
  It is actually simpler to explain the construction of~$\Gamma_1$ first;
  see Figure~\ref{fig:Cycles3} for an example.  For simplicity of notation,
  we let $G_1:=(V,E)$.  Recall that the \emph{cyclomatic number} of a
  connected graph is the minimum number of edges one needs to delete to
  obtain a tree.  Equivalently, it equals its number of edges minus its
  number of vertices plus one.

  \begin{figure}[htb]
    \centering \def\svgwidth{5cm}\begingroup
  \makeatletter
  \providecommand\color[2][]{%
    \errmessage{(Inkscape) Color is used for the text in Inkscape, but the package 'color.sty' is not loaded}
    \renewcommand\color[2][]{}%
  }
  \providecommand\transparent[1]{%
    \errmessage{(Inkscape) Transparency is used (non-zero) for the text in Inkscape, but the package 'transparent.sty' is not loaded}
    \renewcommand\transparent[1]{}%
  }
  \providecommand\rotatebox[2]{#2}
  \ifx\svgwidth\undefined
    \setlength{\unitlength}{1300.575pt}
  \else
    \setlength{\unitlength}{\svgwidth}
  \fi
  \global\let\svgwidth\undefined
  \makeatother
  \begin{picture}(1,0.13115113)%
    \put(0,0){\includegraphics[width=\unitlength]{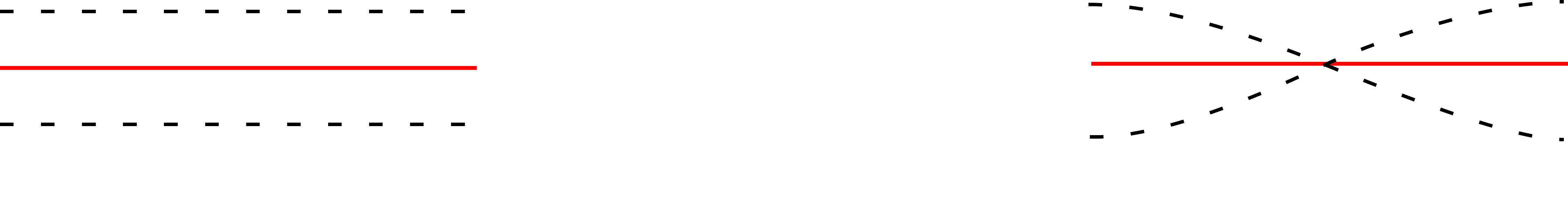}}%
    \put(0.14498532,0.00615113){\color[rgb]{0,0,0}\makebox(0,0)[lb]{\smash{1.}}}%
    \put(0.84181996,0){\color[rgb]{0,0,0}\makebox(0,0)[lb]{\smash{2.}}}%
  \end{picture}%
\endgroup\\
    \vspace{0.5cm}
    
    \includegraphics[width=.4\linewidth]{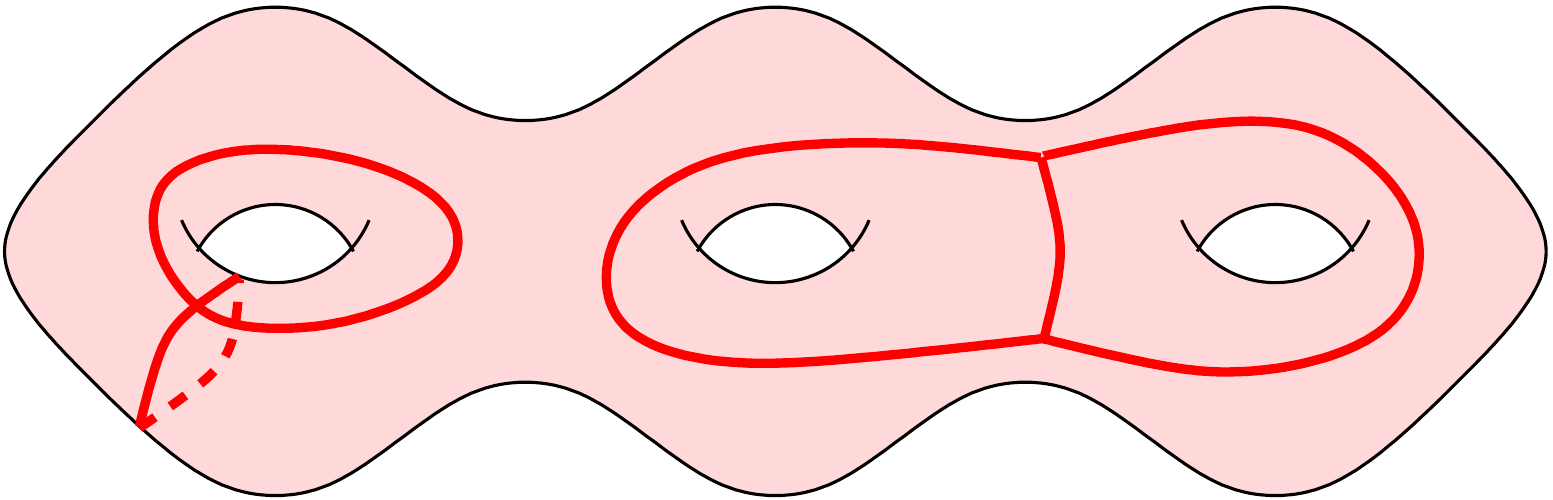}
    \vspace{0.5cm}
    
    \includegraphics[width=.4\linewidth]{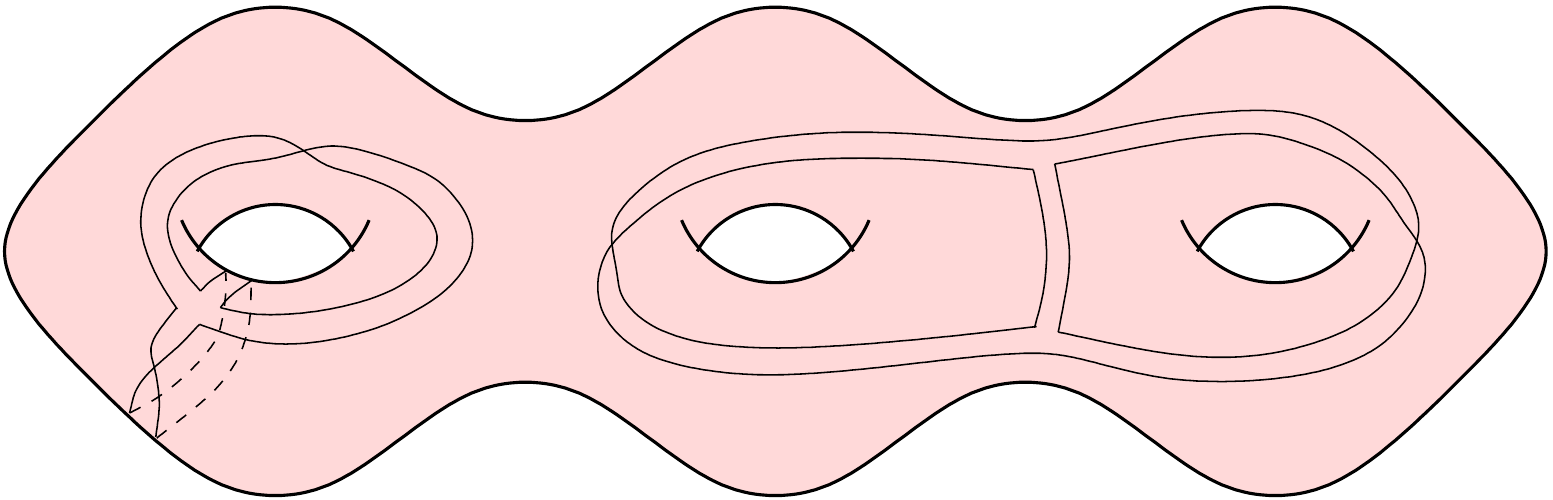}
    
    \caption{Top: A crossover along an edge of $E'\setminus E''$.  Middle:
      An embedded graph~$G_1$ on a genus 3 surface.  Bottom: The
      corresponding family~$\Gamma_1$}
    \label{fig:Cycles3}
  \end{figure}

  Let $(V',E')$ be a connected component of~$(V,E)$.  By
  Proposition~\ref{P:preproc}, we can assume that each vertex has degree at
  least two; in particular, $(V',E')$ has cyclomatic number at least one.
  \begin{itemize}
  \item If $(V',E')$ has cyclomatic number one, then it must be a single
    cycle. In this case, we add that cycle to~$\Gamma_1$.
  \item Otherwise, $(V',E')$ has cyclomatic number at least two.  We add
    to~$\Gamma_1$ the cycles that are the boundaries of a tubular
    neighborhood of~$(V',E')$.  Let $E''$ be the edge set of a spanning
    tree of~$(V',E')$.  For each edge $e\in E'\setminus E''$, we introduce
    a ``crossover'' as in Figure~\ref{fig:Cycles3} (top) on the two pieces
    of~$\Gamma_1$ that run along edge~$e$: Instead of locally having two
    pieces of cycles that run along edge~$e$ without touching it, we now
    have two pieces of cycles in the neighborhood of edge~$e$ that cross at
    a single interior point of~$e$.  Of course, this operation may change
    the number of cycles of~$\Gamma_1$ and create self-intersections.  See
    Figure~\ref{fig:Cycles3}.
  \end{itemize}
  We now prove that $\Gamma_1$ is a stable family. By construction, for
  each connected component, the cycles in $\Gamma_1$ do not intersect the
  chosen spanning tree of~$(V',E')$. Let $f$ be a face of the arrangement
  of the cycles in~$\Gamma_1$.  We have to prove that $f$ is not a $k$-gon
  with $k\le3$. Following the definition of $\Gamma_1$, we observe that $f$
  is either an \emphdef{inner disk}, namely, a disk containing entirely a
  spanning tree $(V',E'')$ of some connected component of~$(V,E)$, or is
  contained entirely in a single face of $(V,E)$.
  \begin{itemize}
  \item Assume first that $f$ is an inner disk, containing the spanning
    tree $(V',E'')$.  By construction, $(V',E')$ has cyclomatic number at
    least two, so $|E'\setminus E''|\ge2$.  Each edge $e\in E'\setminus
    E''$ corresponds to a single crossing between cycles of~$\Gamma_1$, and
    this crossing appears twice along the boundary of~$f$; so $f$ has
    $2|E'\setminus E''|\ge4$ crossings of~$\Gamma$ along its
    boundary.
  \item Otherwise, $f$ has the same topology as a face of~$(V,E)$, and
    therefore cannot be a disk.
  \end{itemize}
  Hence $\Gamma_1$ is a stable family.  It follows from the construction
  that the computation of~$\Gamma_1$ takes linear time in the complexity of
  the combinatorial map of~$G_1$.  The cycles in~$\Gamma_1$ have been
  constructed in a tubular neighborhood of the graph~$G_1$; more precisely,
  by construction, they run along a side of the edges of~$G_1$, swapping
  side whenever they run along an edge not in a spanning tree. Therefore
  (by retracting the tubular neighborhood) they naturally correspond to a
  family of cycles~$\Gamma$ in~$G$, and deducing the family~$\Gamma$
  from~$\Gamma_1$ takes linear time.  All these cycles use each edge of~$G$
  at most twice.  Furthermore, also by construction, $G_1$ is included
  in~$\Gamma_1$ and in the faces of~$\Gamma_1$ that are disks.
\end{proof}

The basic idea of the proof is as follows.  Assume that each cycle
in~$\Gamma_1$ is homotopic to the corresponding cycle in~$\Gamma_2$.
Theorem~\ref{T:stable} implies that, after an ambient isotopy of~$S$, we
can assume that each cycle in~$\Gamma_1$ coincides with the corresponding
cycle in~$\Gamma_2$ not necessarily pointwise, but with the same
orientation.  If this was the case pointwise, then, since $G_i$ is
``surrounded'' by cycles in~$\Gamma_i$, this would imply that $G_1$
and~$G_2$ almost coincide and could be moved one into the other by another
isotopy.  However, the first isotopy does not necessarily map $\Gamma_1$
to~$\Gamma_2$ pointwise, and we need to test that a few more pairs of
cycles are homotopic to ensure that it is the case.

%%%%%%%%%%%%%%%%%%%%%%%%%%%%%%%%%%%%%%%%%%%%%%%%%%%%%%%%%%%%%%%%%%%%%%%%%
\paragraph*{Fixing the Map Automorphism.}\label{S:fix}

\ We now prove:
\begin{proposition}\label{P:fix}
  In linear time, we can construct a family of cycles $\Lambda$ in $G$ such
  that:
\begin{itemize}
\item each edge of $G$ is used at most thrice by all the cycles in
  $\Lambda$.
\item if we denote by $\Lambda_1$ and $\Lambda_2$ the images of $\Lambda$
  in $G_1$ and $G_2$, if every cycle in $\Lambda_1$ is homotopic to its
  counterpart in $\Lambda_2$, then an ambient isotopy of~$S$ maps
  $\Gamma_1$ to $\Gamma_2$ pointwise.
\end{itemize}
\end{proposition}

We will need the following rather independent lemma in the course of the
proof.
\begin{lemma}\label{L:homocycles}
  Let $C$ be a family of simple cycles on~$S$, pairwise disjoint except at
  a single point~$p$, where two cycles may or may not cross.  Assume that
  no component of $S\setminus C$ is a disk bounded by one or two cycles.
  Then the cycles in~$C$ are pairwise (freely) non-homotopic.
\end{lemma}
\begin{proof}
  We will use the fact that two simple homotopic cycles cross transversely
  an even number of times (because they form bigons
  \cite[Proposition~1.7]{fm-pmcg-11}).

  First, no cycle in~$C$ is contractible; otherwise, it would bound a disk
  on the surface.  The cycles inside that disk are all contractible, and
  therefore do not cross at $p$ because of the aforementioned fact.
  Taking an innermost such cycle, we obtain a component of $S\setminus C$
  bounded by one cycle, contradicting the assumption.

  Assume now for the sake of a contradiction that two cycles $c_1$
  and~$c_2$ are homotopic.  They meet at point~$p$ without crossing
  transversely.  After a local perturbation, these cycles become disjoint,
  and therefore bound an annulus \cite[Lemma~2.4]{e-c2mi-66}.  Thus, the
  unperturbed cycles can be viewed as two loops $\ell_1$ and~$\ell_2$ based
  at~$p$ that bound a disk.  There may be other cycles inside that disk,
  but in all cases a face inside it is a disk bounded by one or two cycles,
  which is impossible.
\end{proof}

\begin{proof}[Proof of Proposition~\ref{P:fix}]
  The family $\Lambda$ is the union of the stable family~$\Gamma$ defined
  in Proposition~\ref{P:constr} and of the family~$\Phi$ defined as
  follows.  Recall that in the proof of Proposition~\ref{P:constr}, we
  considered each connected component $(V',E')$ of the graph $G=(V,E)$ in
  turn.  If $(V',E')$ had cyclomatic number at least two, we considered the
  edge set $E''$ of a spanning tree of~$(V',E')$.  A \emphdef{fundamental
    cycle} of~$(V',E')$ is a simple cycle in~$(V',E')$ containing exactly
  one edge in~$E'\setminus E''$.  We put in~$\Phi$ an arbitrary fundamental
  cycle for each connected component $(V',E')$ of cyclomatic number at
  least two.  The fundamental cycles of a connected component $(V',E')$ can
  be extended towards an arbitrary root~$p$ of the spanning tree~$(V',E'')$
  and then slightly perturbed on~$S$ so that they become simple and
  pairwise disjoint except at~$p$, where they may or may not cross.  The
  faces of this new family~$C$ of perturbed cycles correspond to the faces
  of~$(V',E')$.  Moreover, $C$ satisfies the hypotheses of
  Lemma~\ref{L:homocycles}: Indeed, if there is a disk in $S \setminus C$
  bounded by one or two cycles, there must be at least one connected
  component of $G_1$ inside it because no face of $G_1$ is a disk; but then
  this connected component is contractible, which is absurd since the
  preprocessing removed all the contractible components of $G_1$.  Hence,
  the fundamental cycles of any given connected component $(V',E')$ are
  pairwise non-homotopic.

  Clearly the family $\Lambda=\Gamma\cup\Phi$ can be computed in linear
  time and uses each edge of~$G$ at most thrice.  Assume that, for each
  cycle $\lambda$ in~$\Lambda$, the images of~$\lambda$ in $G_1$ and~$G_2$
  are homotopic.  There remains to prove that some isotopy of~$S$ maps
  $\Gamma_1$ to~$\Gamma_2$ pointwise.

  By Proposition~\ref{P:constr}, $\Gamma_1$ is a stable family, and of
  course $\Gamma_2:=h(\Gamma_1)$ as well.  Since each cycle in~$\Gamma_1$
  is homotopic to the corresponding cycle in~$\Gamma_2$,
  Theorem~\ref{T:stable} implies that some isotopy of~$S$ takes $\Gamma_1$
  to~$\Gamma_2$, not necessarily pointwise, but preserving the orientations
  of the cycles.%
  \footnote{Actually, if $S$ has nonnegative Euler characteristic, the
    results in Section~\ref{S:exc} show that the isotopy can be chosen so
    as to be pointwise, which concludes the proof of this proposition.} %
  Therefore, up to composing with this isotopy, we can assume that each
  cycle in~$\Gamma_1$ coincides, as a set, with the corresponding cycle
  in~$\Gamma_2$, and with the same orientation.  Hence, this isotopy
  induces an orientation-preserving map isomorphism~$i$ between $\Gamma_1$
  and $\Gamma_2$, and since $\Gamma_1$ and $\Gamma_2$ have the same
  extended combinatorial maps, $i$ can be viewed as an
  orientation-preserving map automorphism of $\Gamma_1$.  As each cycle in
  $\Gamma_1$ is isotoped to the corresponding cycle in $\Gamma_2$, $i$ maps
  each connected component of $\Gamma_1$ to itself, and it maps each
  crossing of~$\Gamma_1$ to a crossing of~$\Gamma_1$.

  We now want to ensure that $i$ is the identity map automorphism, which
  would imply the existence of a pointwise ambient isotopy between
  $\Gamma_1$ and $\Gamma_2$ (as in the proof of Corollary~\ref{C:stable}).
  However, this is not necessarily the case, as was pictured in
  Figure~\ref{F:contreex2}.

  Let $(V',E')$ be a connected component of~$G$; let $\Gamma'_1$ be the
  arrangement of the cycles of~$\Gamma_1$ corresponding to that connected
  component.  If $(V',E')$ has cyclomatic number one, by construction,
  $\Gamma'_1$ is just a cycle, which $i$ maps to itself, preserving its
  orientation; so $i$ is the identity map automorphism on~$\Gamma'_1$.

\begin{figure}[htb]
\centering
\def\svgwidth{7cm}
\begingroup
  \makeatletter
  \providecommand\color[2][]{%
    \errmessage{(Inkscape) Color is used for the text in Inkscape, but the package 'color.sty' is not loaded}
    \renewcommand\color[2][]{}%
  }
  \providecommand\transparent[1]{%
    \errmessage{(Inkscape) Transparency is used (non-zero) for the text in Inkscape, but the package 'transparent.sty' is not loaded}
    \renewcommand\transparent[1]{}%
  }
  \providecommand\rotatebox[2]{#2}
  \ifx\svgwidth\undefined
    \setlength{\unitlength}{2016.6301551pt}
  \else
    \setlength{\unitlength}{\svgwidth}
  \fi
  \global\let\svgwidth\undefined
  \makeatother
  \begin{picture}(1,1.00017618)%
    \put(0,0){\includegraphics[width=\unitlength]{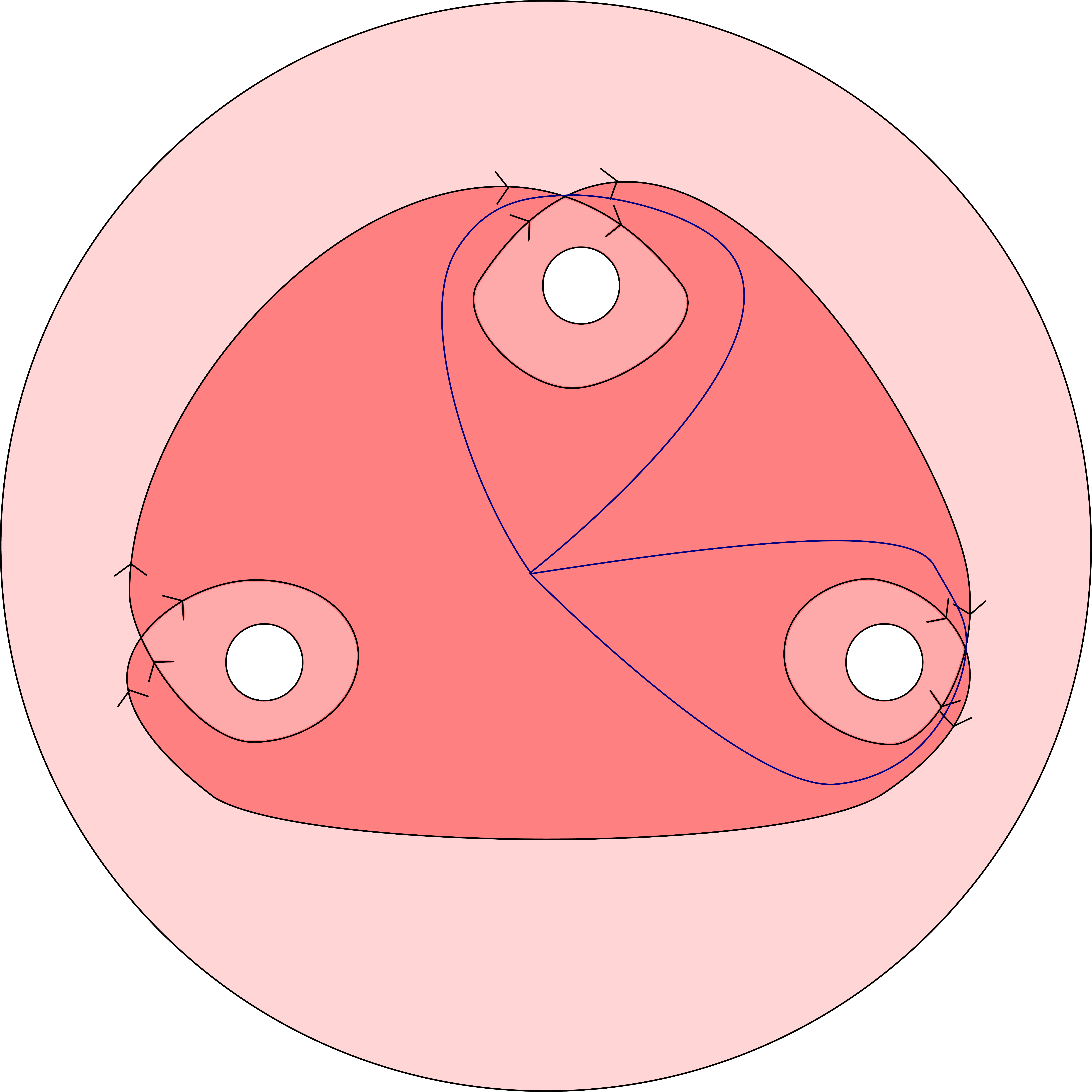}}%
    \put(0.64902127,0.82933445){\color[rgb]{0,0,0}\makebox(0,0)[lb]{\smash{$\Gamma_1$}}}%
    \put(0.51300936,0.84180221){\color[rgb]{0,0,0}\makebox(0,0)[lb]{\smash{$v$}}}%
    \put(0.9074439,0.39522978){\color[rgb]{0,0,0}\makebox(0,0)[lb]{\smash{$v'$}}}%
    \put(0.34186105,0.59018019){\color[rgb]{0,0,0}\makebox(0,0)[lb]{\smash{$\varphi$}}}%
    \put(0.67963777,0.52116281){\color[rgb]{0,0,0}\makebox(0,0)[lb]{\smash{$\varphi'$}}}%
  \end{picture}%
\endgroup
\caption{A graph $\Gamma_1$ drawn on a sphere with four holes. The inner disk
  is the darker part, and the orientations at each vertex are pictured
  according to the orientation of the single cycle. If $i(v)=v'$, $\varphi$
  is sent to $\varphi'$ which is not homotopic to it.}
\label{fig:fix}
\end{figure}
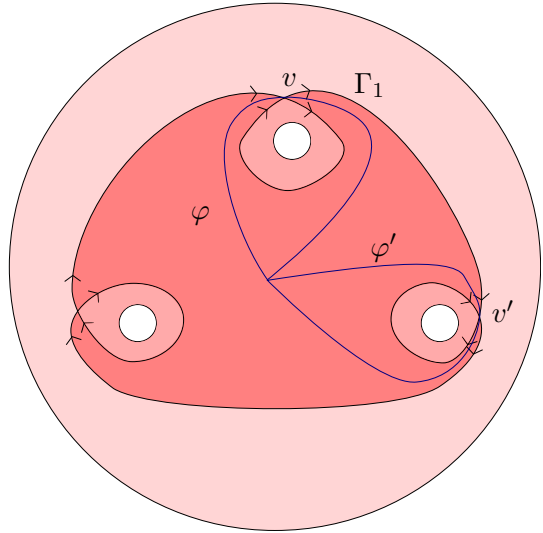

Otherwise, $(V',E')$ has cyclomatic number at least two.  We first note
that $\Gamma'_1$ is connected; indeed, the inner disk~$D$ of~$(V',E')$ is
bounded by all cycles in~$\Gamma'_1$.  Moreover, the faces of $\Gamma_1$
that are disks are exactly the inner disks; so $i$ maps inner disks to
inner disks, and therefore maps $D$ to itself.

Let $v$ be the vertex of~$\Gamma'_1$ corresponding to the cycle~$\varphi$
in~$\Phi=\Lambda\setminus\Gamma$ (see Figure~\ref{fig:fix}).  If $i$ maps
$v$ to another vertex $v'$ of~$\Gamma'_1$, as the inner disk is mapped to
itself, $i$ necessarily maps $\varphi$ to a cycle in the inner disk
crossing $v'$ once, i.e. another fundamental cycle in~$(V',E')$, which is,
as shown above, not homotopic to~$\varphi$; this is a contradiction.  So
$i$ maps~$v$ to itself.  Furthermore, if we orient the four edges incident
to~$v$ with the orientation of the corresponding cycles, $v$ has two
outgoing edges, consecutive in the cyclic order around~$v$, and two
incoming edges, also consecutive.  Since $i$ maps each edge to another edge
with the same orientation, it maps~$v$ to~$v$, and it is an
orientation-preserving map automorphism, it must thus map each edge
incident to~$v$ to itself, with the same orientation.  Since $\Gamma'_1$ is
connected, by propagation we deduce that $i$ is the identity map
automorphism on~$\Gamma'_1$, which concludes the proof.
\end{proof}

%%%%%%%%%%%%%%%%%%%%%%%%%%%%%%%%%%%%%%%%%%%%%%%%%%%%%%%%%%%%%%%%%%%%%%%%%%
\paragraph*{End of Proof of Theorem~\ref{T:ladeg}.}

\ We now conclude the proof of Theorem~\ref{T:ladeg} if none of the faces of
$G_1$ are disks.

According to the hypotheses, for all the cycles $\gamma \in \Lambda$, the
images of $\gamma$ in $G_1$ and $G_2$ are homotopic, which implies by
Proposition~\ref{P:fix} that we can assume that $\Gamma_1=\Gamma_2$
pointwise. Then, each face of $\Gamma_1$ is mapped by $h$ to itself,
because~$h$ is an oriented homeomorphism.

In particular, $G_1\cap\Gamma_1=G_2\cap\Gamma_2$.  In every disk of
$S\setminus\Gamma_1=S\setminus\Gamma_2$, the oriented homeomorphism~$h$ is
the identity on the boundary; therefore, by Alexander's lemma, it is an
ambient isotopy relatively to the boundary.  This gives us an isotopy
between $G_1$ and $G_2$ on every such disk, relatively to
$\Gamma_1=\Gamma_2$.  By gluing these isotopies together along their
boundaries, we get an isotopy of~$S$ mapping~$G_1$ to~$G_2$, because $G_1$
and~$G_2$ are included in the closures of the faces of~$\Gamma_1=\Gamma_2$
that are disks (Proposition~\ref{P:constr}).

Since the family $\Lambda$ covers each edge of~$G$ at most thrice, it has
linear complexity. As it can be computed in linear time, this concludes the
proof of Theorem~\ref{T:ladeg} if none of the faces of $G_1$ are disks.

%%%%%%%%%%%%%%%%%%%%%%%%%%%%%%%%%%%%%%%%%%%%%%%%%%%%%%%%%%%%%%%%%%%%%%%%%
\subsection{Proof of Theorem~\ref{T:ladeg} if the only face of $G_1$ is a
  disk}\label{S:cutgraph}

By Proposition~\ref{P:preproc}, either (1) $G_1$ has no face that is a
disk, or (2) $G_1$ has a single face, and that face is a disk.  We proved
Theorem~\ref{T:ladeg} in case (1) in the previous section, and shall now
deal with case (2).  In other words, we assume that $G_1$ is a \emphdef{cut
  graph}.

In that case, the above construction does not seem to work: Since inner
disks are not the only faces of $\Gamma_1$ that are disks, there is no
guarantee that an inner disk is mapped to itself in the proof of
Proposition~\ref{P:fix}.  To circumvent this issue, the high-level idea is
the following: We remove one cycle from $G_1$ so that the only face of
$G_1$ is not a disk anymore but a cylinder, in which case the results from
the previous section apply.  We then check that the remaining cycle and its
counterpart in $G_2$ are homotopic, and prove that this guarantees the
existence of an isotopy between $G_1$ and $G_2$.

Since $G_1$ is a cut graph, $S$ has no boundary.  Moreover, $G$ is
connected, and $G_1$ is made of a spanning tree $T=(V',E')$ (as in section
\ref{S:construction}) and $2g$ additional edges.  Let $e$ be one of these
edges (chosen arbitrarily); let $\gamma_e$ be the fundamental cycle with
respect to~$T$ corresponding to edge~$e$.  Let $G'$ be the graph $G$ with
edge~$e$ removed, and let $G'_1$ and~$G'_2$ be the restrictions of $G_1$
and~$G_2$ to~$G'$.  Note that $G_1'$ has a single face, which is a
cylinder.

We can now apply the result of the previous section to $G'_1$ and~$G'_2$:
We obtain a family $\Lambda'$ of cycles in~$G'$ with the property that, if
their images in $G'_1$ and $G'_2$ are homotopic, then $G'_1$ and~$G'_2$ are isotopic.
Furthermore, $\Lambda'$ can be computed in linear time, and uses each edge
of~$G'$ at most thrice.

To prove Theorem~\ref{T:ladeg} for our graph~$G$, we take
$\Lambda:=\Lambda'\cup\{\gamma_e\}$.  Obviously, $\Lambda$ can be computed
in linear time and uses each edge of~$G$ at most four times in total.
Assume now that the images of each cycle of~$\Lambda$ in $G_1$ and~$G_2$
are homotopic.  It suffices to prove that, under this condition, some
ambient isotopy of~$S$ takes $G_1$ to~$G_2$.  Since the images of each
cycle in~$\Lambda'$ in $G'_1$ and~$G'_2$ are homotopic, we may assume that
$G'_1=G'_2$.  There remains to prove that an isotopy of the surface allows
to push the image~$e_1$ of~$e$ in~$G_1$ to the image~$e_2$ of~$e$ in~$G_2$.

The surface obtained after cutting~$S$ along $G'_1=G'_2$ is a cylinder $C$,
and the images of $e$ on $C$ become arcs $a_1$ and $a_2$ with the same
endpoints, one on each boundary, as shown on
Figure~\ref{fig:cutgraph}. (Indeed, if both endpoints were on the same
boundary, $a_1$ would bound two faces on $C$ which would correspond to two
faces on $S$, reaching a contradiction.)  Now, to conclude, we only need to
show that $a_1$ and $a_2$ are isotopic relatively to the boundary of this
cylinder.
  \begin{figure}[htb]
    \centering \def\svgwidth{15cm}
    \begingroup
  \makeatletter
  \providecommand\color[2][]{%
    \errmessage{(Inkscape) Color is used for the text in Inkscape, but the package 'color.sty' is not loaded}
    \renewcommand\color[2][]{}%
  }
  \providecommand\transparent[1]{%
    \errmessage{(Inkscape) Transparency is used (non-zero) for the text in Inkscape, but the package 'transparent.sty' is not loaded}
    \renewcommand\transparent[1]{}%
  }
  \providecommand\rotatebox[2]{#2}
  \ifx\svgwidth\undefined
    \setlength{\unitlength}{3074.44644pt}
  \else
    \setlength{\unitlength}{\svgwidth}
  \fi
  \global\let\svgwidth\undefined
  \makeatother
  \begin{picture}(1,0.23000222)%
    \put(0,0){\includegraphics[width=\unitlength]{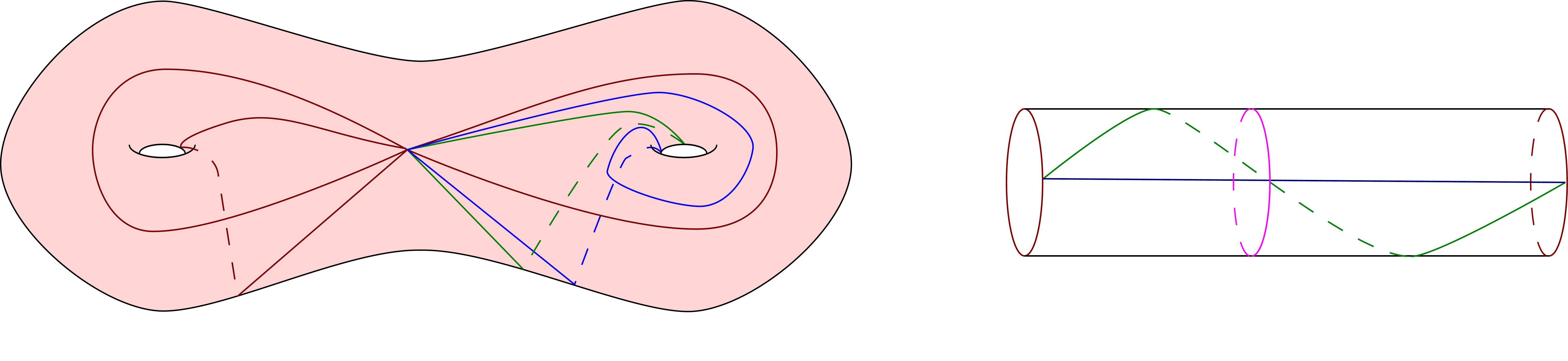}}%
    \put(0.26474846,0.00006861){\color[rgb]{0,0,0}\makebox(0,0)[lb]{\smash{1.}}}%
    \put(0.82304451,-0){\color[rgb]{0,0,0}\makebox(0,0)[lb]{\smash{2.}}}%
    \put(0.10283209,0.19657839){\color[rgb]{0,0,0}\makebox(0,0)[lb]{\smash{$G'_1$}}}%
    \put(0.45189853,0.16713906){\color[rgb]{0,0,0}\makebox(0,0)[lb]{\smash{$e_1$}}}%
    \put(0.42456201,0.15662501){\color[rgb]{0,0,0}\makebox(0,0)[lb]{\smash{$e_2$}}}%
    \put(0.72736662,0.12192864){\color[rgb]{0,0,0}\makebox(0,0)[lb]{\smash{$a_1$}}}%
    \put(0.90505406,0.08512948){\color[rgb]{0,0,0}\makebox(0,0)[lb]{\smash{$a_2$}}}%
    \put(0.6432542,0.17791595){\color[rgb]{0,0,0}\makebox(0,0)[lb]{\smash{$b$}}}%
    \put(0.78729668,0.17791595){\color[rgb]{0,0,0}\makebox(0,0)[lb]{\smash{$b'$}}}%
  \end{picture}%
\endgroup
    \caption{1. The graph $G'_1=G'_2$ and the edges $e_1$ and $e_2$. 2. The
      cylinder $C$ obtained after cutting along $G'_1$, with the arcs $a_1$
      and $a_2$ corresponding to $e_1$ and $e_2$, and the cycles $b$ and
      $b'$. If $a_1$ and $a_2$ are not isotopic, the images of~$\gamma_e$
      in $G_1$ and~$G_2$ cannot be (freely) homotopic}
    \label{fig:cutgraph}
  \end{figure}
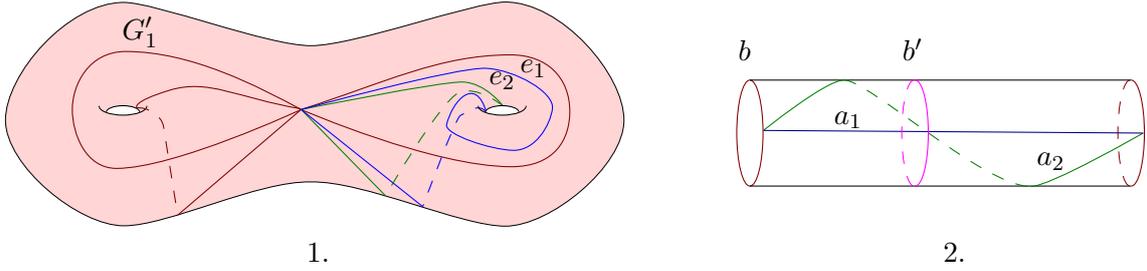

  The end of the proof uses some elementary notions of homology, we refer
  to \cite{h-at-02} for the background. Assume, for the sake of a
  contradiction, that $a_1$ and~$a_2$ are non-isotopic arcs relatively to
  the boundary of~$C$.  This implies that they are non-homotopic
  on~$C$~\cite[Theorem~3.1]{e-c2mi-66}.  Hence there exists an
  integer~$n\ne0$ such that $a_2$ is homotopic to $a_1\cdot b^n$, where $b$
  is a loop that is a boundary of the cylinder~$C$.  Since the images
  of~$\gamma_e$ in $G_1$ and~$G_2$ are (freely) homotopic on~$S$, they are
  $\Z$-homologous.  This implies that $b^n$, and thus~$b$, has zero
  $\Z$-homology.  By translating along the cylinder, $b$ is homotopic
  on~$S$ to a simple cycle $b'$ that crosses $e_1$ exactly once and crosses
  $G_1$ nowhere else, as pictured on Figure~\ref{fig:cutgraph}. Hence $b'$
  is a simple cycle on~$S$ that crosses the image of $\gamma_e$ in~$G_1$
  exactly once; thus $b'$ is non-separating, and therefore cannot have zero
  $\Z$-homology.  This contradiction completes the proof of
  Theorem~\ref{T:ladeg}.

%%%%%%%%%%%%%%%%%%%%%%%%%%%%%%%%%%%%%%%%%%%%%%%%%%%%%%%%%%%%%%%%%%%%%%%%%%
\section{Algorithms}\label{S:main-thms}

For Theorems \ref{T:main-surf} and~\ref{T:main-plane}, it suffices to be
able to test the existence of an oriented homeomorphism, and of homotopies
between the cycles in~$\Lambda$, as computed by Theorem~\ref{T:ladeg}, in
the indicated amount of time.  We prove Theorems \ref{T:main-surf}
and~\ref{T:main-plane} in Sections \ref{S:main-surf}
and~\ref{S:main-plane}, respectively.

%%%%%%%%%%%%%%%%%%%%%%%%%%%%%%%%%%%%%%%%%%%%%%%%%%%%%%%%%%%%%%%%%%%%%%%%%%
\subsection{Surfaces: Proof of Theorem~\ref{T:main-surf}}\label{S:main-surf}

Recall that, in Theorem~\ref{T:main-surf}, the input of the algorithm
consists of a fixed graph~$H$ cellularly embedded on a fixed surface~$S$,
and of embeddings $G_1$ and~$G_2$ of a graph~$G$.  Furthermore, $k_1$
(resp.,~$k_2$) denotes the complexity of the combinatorial map of the
arrangement of~$G_1$ (resp.,~$G_2$) with~$H$.

\emph{Homeomorphism test.}  For the case of graphs on surfaces, the
existence of an oriented homeomorphism that maps $G_1$ to~$G_2$ can be
checked in $O(k_1+k_2)$ time.  Indeed, let us choose an arbitrary
orientation on~$S$; this induces an orientation of the combinatorial map
of~$H$, and hence an orientation of the combinatorial map of the arrangement of
$G_i$ and~$H$, for $i=1,2$.  Computing the number of boundary components
of~$S$ in each face, as well as the genus of each face (using the Euler
characteristic), and ``erasing'' the graph~$H$ in both arrangements gives
us oriented combinatorial maps for each~$G_i$.  They are isomorphic if and
only if there exists an oriented isomorphism of~$S$ between $G_1$ and~$G_2$
(Lemma~\ref{L:extcombtest}), and this can be checked in linear time.

\medskip

\emph{Homotopy tests.}  The homotopy tests can also be performed in
$O(k_1+k_2)$ time.  Indeed, recall that the input to the algorithm consists
of the combinatorial maps of the arrangement of $G_1$ and~$H$ on one hand,
and of $G_2$ and~$H$ on the other hand, where $H$ is a fixed cellular graph
embedding; $k_1$ and~$k_2$ denote the complexities of these two maps.  In
$O(k_1+k_2)$ total time, we can compute the cyclically ordered list of
edges of~$H$ crossed by each cycle of~$\Lambda$ in the embeddings $G_1$
and~$G_2$.  This gives us a set of pairs of cycles in the dual graph~$H^*$
of~$H$ that have to be tested for homotopy.  The total complexity of these
cycles is $O(k_1+k_2)$, and $H^*$ has complexity $O(k_1+k_2)$ as well.
Lazarus and Rivaud, and later Erickson and
Whittlesey~\cite{lr-hts-12,ew-tcsr-13} prove that, after a preprocessing
linear in the complexity of the cellular graph~$H^*$, one can test homotopy
of cycles in~$H^*$ in time linear in the complexities of these cycles. (An
earlier paper by Dey and Guha~\cite{dg-tcs-99} claims a similar result,
except for some low-genus surfaces, but Lazarus and Rivaud point out some
problems in their proof.)  These papers address only the case of surfaces
without boundary, but the case of surfaces with boundary is easier, as the
fundamental group is free; alternatively, homotopy tests for cycles on
surfaces with boundary can be performed using an algorithm for surfaces
without boundary by first attaching a handle to each boundary component,
which does not change the outcomes of the homotopy tests.

This concludes the proof of Theorem~\ref{T:main-surf}.

%%%%%%%%%%%%%%%%%%%%%%%%%%%%%%%%%%%%%%%%%%%%%%%%%%%%%%%%%%%%%%%%%%%%%%%%%%
\subsection{Punctured Plane: Proof of
  Theorem~\ref{T:main-plane}}\label{S:main-plane}

We now give our algorithm for the punctured plane model; so let $G_1$
and~$G_2$ be two embeddings of a graph~$G$ in the punctured
plane~$\R^2\setminus P$.  Let $\cal P$ be a set of disjoint open polygons
(for example squares), one around each point of~$P$, that avoid $G_1$
and~$G_2$; also, let $B$ be a large closed square such that $G_1$, $G_2$,
and the closure of~$\cal P$, are in the interior of~$B$.  By compactness,
any isotopy between $G_1$ and~$G_2$, if it exists, must avoid neighborhoods
of~$P$ and stay in a bounded area of the plane; therefore, such an isotopy
exists if and only if such an isotopy exists in $B\setminus\cal P$.  In
other words, since $B\setminus\cal P$ is a surface with boundary, we are
exactly in the topological setting of the previous sections, except that
the input to the algorithm is given in a different form.

As above, our algorithm relies on two subroutines: a test for the existence
of an oriented homeomorphism, and a test for homotopy between cycles.  We
actually give two algorithms for the latter problem, because, depending on
the ratio between $k_1+k_2$ and~$p$, one is faster than the other.

\medskip

\emph{Homeomorphism test}.  To test whether there exists an oriented
homeomorphism of the plane that maps $G_1$ to~$G_2$ in~$\R^2\setminus P$,
we compute the oriented combinatorial map of~$G_1$ in the punctured plane
using a sweep-line algorithm for $G_1\cup P$ in $O((k_1+p)\log(k_1+p))$
time~\cite{bo-arcgi-79}.  Then we apply the same procedure with~$G_2$
instead of~$G_1$, and check that the two resulting oriented combinatorial
maps are isomorphic.

\medskip

\emph{First algorithm for homotopy tests.}  We transform the input into the
surface model.  For this purpose, we compute a triangulation~$T$ of~$B$ in
$O(p\log p)$ time (for example, a Delaunay triangulation).  We can then
easily determine the arrangement of $G_1$ with~$T$ in $O(k_1p)$ time,
because each segment in~$G_1$ has $O(p)$ crossings with~$T$; that
arrangement has complexity $O(k_1p)$.  We can apply the same procedure
to~$G_2$.  After a slight modification of~$T$ that does not affect its
complexity, we may assume that $T$~is a triangulation of the bounding box
minus a set of small square obstacles.  We can then test homotopy of cycles
in time linear in the number of their crossings with~$T$, either by
computing and comparing their cyclically reduced crossing words with~$T$
(since the fundamental group is a free group) or by applying the algorithm
by Lazarus and Rivaud or the one by Erickson and Whittlesey.  This takes
$O((k_1+k_2)p)$ time.

\medskip

\emph{Second algorithm for homotopy tests.}  To get a subquadratic running
time in the input size, we improve the homotopy test by adapting an
algorithm by Cabello et al.~\cite[Section~4]{clms-thpp-04}.  Their
algorithm tests homotopy for paths in the punctured plane, not homotopy for
cycles; however, it can be modified to handle this case also.  More
precisely, we show below that, after $O(p^{1+\varepsilon})$ preprocessing
time (for any~$\varepsilon>0$), one can test homotopy of two (possibly
non-simple) cycles $\gamma_1$ and~$\gamma_2$ of complexities $m_1$
and~$m_2$, respectively, in $O((m_1+m_2)\sqrt{p}\log p)$ time.

The main idea is to replace the triangulation~$T$ in the first algorithm
above with a cellular decomposition of~$B$ that has a nicer property: Each
line in the plane crosses at most $O(\sqrt{p})$ segments of that
decomposition.  Cabello et al.~\cite{clms-thpp-04} show how to compute such
a cellular decomposition in $O(p^{1+\varepsilon})$ time.  The two input
cycles $\gamma_1$ and~$\gamma_2$ cross this decomposition
$O((m_1+m_2)\sqrt{p})$ times.  Then, computing the cyclically ordered lists
of edges of the decomposition crossed by these two cycles takes
$O((m_1+m_2)\sqrt{p}\log p)$ time using ray shooting, as done also in the
paper by Cabello et al.  We conclude using the same method as in the first
algorithm, with the cellular decomposition in place of the
triangulation. This proves Theorem~\ref{T:main-plane}.

%%%%%%%%%%%%%%%%%%%%%%%%%%%%%%%%%%%%%%%%%%%%%%%%%%%%%%%%%%%%%%%%%%%%%%%%%%
\section{Graph isotopies with fixed vertices}\label{S:fixed}

In this section, we briefly indicate how the previous techniques extend to
the graph isotopy problem where, in addition, we require the isotopy to fix
some vertices.  Formally, for $G_1$ and $G_2$ two embeddings of a graph
$G=(V,E)$ on the interior of a surface~$S$ and a set $V_f \subseteq V$ such
that the embeddings of $V_f$ are the same in both graphs, we want to test
whether there exists an isotopy $h_t$ between $G_1$ and $G_2$ such that
$h_t|_{V_f}$ is the identity for all $t \in [0,1]$. We call this the
\emph{fixed vertices graph isotopy problem}.

We will deal with this variant by applying our algorithm in the setting of
surfaces with punctures. Surfaces with punctures are surfaces where we
removed a finite number of points. We will denote such a surface by
$(S,P)$, where $P$ is the set of punctures on the surface without
punctures~$S$. Although such punctured surfaces are not compact, they share
many properties with usual surfaces. In particular, they can be endowed
with a hyperbolic metric if their Euler characteristic is negative, where
we define it by $\chi(S,P)=2-2g-b-|P|$.

\begin{theorem}[{\cite[Theorem~1.2]{fm-pmcg-11}}]
  Consider any surface, perhaps with punctures or boundary. If its Euler
  characteristic is negative, then there exists a complete, finite-area
  Riemannian metric of constant curvature $-1$ on the surface, with the
  property that the boundary of the surface is totally geodesic.
\end{theorem}

Hyperbolic surfaces with punctures share many properties with usual
hyperbolic surfaces; we refer the reader to \cite{fm-pmcg-11} for details.
We will mainly use the following proposition, which we give without proof.
A cycle~$\gamma$ is \emphdef{homotopic into a neighborhood of a puncture}
if, for every neighborhood~$N$ of that puncture, there is a cycle homotopic
to~$\gamma$ that lies entirely within~$N$.
\begin{proposition}[{\cite[Proposition~1.3]{fm-pmcg-11}}]\label{propfm}
  Let $S$ be a hyperbolic surface. If $\alpha$ is a cycle in $S$ that is
  not homotopic into a neighborhood of a puncture, then $\alpha$ is
  homotopic to a unique geodesic cycle~$\gamma$.
\end{proposition}

In a nutshell, the strategy for solving the fixed vertices graph isotopy
problem is to put punctures on the fixed vertices of~$G_1$ (or,
equivalently, $G_2$).  We essentially use the same algorithm as in the
compact case, with some minor modifications listed below.  To keep things
simple, we only deal with the case where the resulting surface with
punctures is hyperbolic; the remaining cases (the once or twice punctured
sphere and the once punctured disk) are simpler and can be dealt with the
same idea coupled with the techniques of Section~\ref{S:exc}.

\paragraph{Isotopies of stables families of cycles.} \ The analogue of
Theorem~\ref{T:stable} still holds in the case of surfaces with punctures:
\begin{theorem}\label{T:stablepunc}
  Let $(S,P)$ be an orientable surface with punctures and let
  $\Gamma_1=(\gamma_{1,1}, \ldots, \gamma_{1,n})$ and
  $\Gamma_2=(\gamma_{2,1}, \ldots \gamma_{2,n})$ be two \emph{stable}
  families of cycles on~$(S,P)$ in general position such that:
  \begin{enumerate}
  \item there exists an oriented homeomorphism~$h$ of $(S,P)$ mapping each
    cycle~$\gamma_{1,j}$ of~$\Gamma_1$ to the corresponding
    cycle~$\gamma_{2,j}$ of~$\Gamma_2$ not necessarily pointwise, but
    preserving the orientations of the cycles, and
  \item each cycle of~$\Gamma_1$ is homotopic to the corresponding cycle
    of~$\Gamma_2$.
  \end{enumerate}
  Then there is an isotopy of $(S,P)$ that maps each cycle of~$\Gamma_1$ to the
  corresponding cycle of~$\Gamma_2$, not necessarily pointwise, but
  preserving the orientations of the cycles.
\end{theorem}
\begin{proof}
  The proof is essentially the same as for Theorem~\ref{T:stable} (in the
  hyperbolic case). Indeed, all the properties related to hyperbolicity are
  maintained, with a small caveat stated in Proposition~\ref{propfm}: For
  the cycles that are homotopic into the neighborhood of a puncture, there is
  no corresponding geodesic, as it would be infinitely small. But when
  we do local shortenings as in Proposition~\ref{proposition13}, the cycles
  homotopic to the neighborhood of a puncture are pushed into an
  $\varepsilon$-neighborhood of this puncture, and this allows us to define
  an \textit{annular corridor} around it, where, using the same proof as
  Proposition~\ref{P:trgle}, the cycles do not cross. Since each cycle
  $\gamma$ of $\Gamma_1$ is homotopic to $h(\gamma)$, and $h$ preserves the
  orientations of the cycles, each puncture with an annular corridor is
  mapped to itself by $h$. In this annular neighborhood, the orderings of
  Section~\ref{S:tech} are the same because the oriented homeomorphism $h$
  maps the puncture to itself and can thus be extended to the annular
  corridor. The rest of the proof follows identically.
\end{proof}

\paragraph{Preprocessing.} 
\ On a surface with punctures, we extend slightly our definition of graph
embeddings to allow the graphs to hit punctures. Formally, a graph
$G=(V,E)$ is embedded on $(S,P)$ in general position if it is embedded on
$S$ and $G \cap P\subseteq V$. By a slight abuse of language, the
\textit{faces} of such an embedding will denote the faces of the usual
embedding we obtain by removing the punctures under the vertices. The
following lemma is straightforward.

\begin{lemma}
  Let $G=(V,E)$ be embedded as $G_1$ and $G_2$ on a surface $S$, such that
  both embeddings are identical for $V_f \subseteq V$. Then $G_1$ and $G_2$
  are isotopic with $V_f$ fixed if and only if $G_1$ and $G_2$ (seen as
  embeddings on $(S,V_f)$) are isotopic in $(S,V_f)$.
\end{lemma}

\begin{proof}
  The same isotopy can be used in both settings.
\end{proof}
Thus, solving the fixed vertices graph isotopy problem amounts to solving
the graph isotopy problem on surfaces with punctures, where the punctures
are placed on the fixed vertices.

The preprocessing step is slightly altered. We obviously do not want to
remove from the graphs degree-one vertices that are supposed to stay fixed
during the isotopy, so we use the following proposition instead.

\begin{proposition}\label{P:preproc2}
  Without loss of generality, we may assume that (1) $G_1$ has a single
  face, or none of its faces is a disk, (2) $G_1$ has no vertices of degree
  zero, and (3) The only vertices of $G_1$ of degree one are on punctures.
\end{proposition}

The proof is the same, except that we do not touch the fixed vertices of
degree one, i.e., vertices of degree one on punctures.

\paragraph{Construction of the stable family.} \ Unlike the compact case, we
do not distinguish two cases in the construction of the stable families $\Gamma_1$ and~$\Gamma_2$. It follows closely the one described in
Proposition~\ref{P:constr}, but we need to ensure that there is at most one
puncture in the disks we obtain.

As in the proof of Proposition~\ref{P:constr}, we build the stable
family~$\Gamma_1$ corresponding to~$G_1$ and then map it to~$\Gamma_2$ by
the homeomorphism that maps $G_1$ to~$G_2$.  To build~$\Gamma_1$, we
consider each connected component of~$G_1$ in turn.  Let $(V',E')$ be such
a connected component and let $V'_f=V'\cap V_f$ be the set of vertices on
the punctures.  If $V'_f=\emptyset$, we apply the same procedure as in the
proof of Proposition~\ref{P:constr}.  Otherwise, we compute a spanning
forest $(V',E'')$ of~$(V',E')$, where each tree of the forest contains
exactly one vertex of~$V'_f$.  As in the proof of
Proposition~\ref{P:constr}, we add to the stable family the cycles that are
the boundaries of a tubular neighborhood of~$(V',E')$, but adding a
``crossover'' at each edge in $E'\setminus E''$.

To see that this results in a stable family~$\Gamma_1$, note that (as
before) the faces of the family are of two types:
\begin{itemize}
\item \emph{inner disks}, which contain some pieces of the graph~$G$; such
  faces either contain exactly one puncture in their interior (in case they
  are built as indicated in the previous paragraph), or contain no puncture
  but have degree at least four (as in the proof of
  Proposition~\ref{P:constr});
\item \emph{outer disks}, which have the topology of the faces of the
  graph~$G_1$, and thus, by the preprocessing step, cannot be disks, unless
  there is exactly one such disk, which must therefore be a cut graph of
  the surface~$S$ without boundary (and after disregarding the punctures).
  Such a disk has degree at least four if the surface has genus at least
  one; if the surface~$S$ is a sphere, the hyperbolicity hypothesis implies
  that there are at least three punctures, making at least two crossovers
  and thus an outer face of degree at least four.
\end{itemize}

\paragraph{Fixing the map automorphism.} \ We now need to prove the
counterpart of Proposition~\ref{P:fix}.  Let $(V',E')$ be a connected
component of~$G_1$.  If $V'$ contains no puncture, then we proceed as in
the proof of Proposition~\ref{P:fix}: After testing the homotopy class of a
given fundamental cycle, we can certify that the map automorphism is fixed
(or, in other words, that some ambient isotopy maps the stable cycles
corresponding to~$(V',E')$ \emph{pointwise}).  We now consider the case
where $V'$ contains at least one puncture.

Since each puncture is mapped to itself using the map automorphism, we know
that each inner disk of $(V',E')$ is mapped to itself.  We use an argument
similar to the proof of Proposition~\ref{P:fix}, but the counterpart of
Lemma~\ref{L:homocycles} is simpler in our case.

We choose an arbitrary \emphdef{fundamental path} in~$(V',E')$, which
connects two punctures (possibly identical) using exactly one edge
\emph{not} in the spanning forest $(V',E'')$.  (Such an edge exists, since,
after the preprocessing step, it cannot be that $(V',E')$ is a tree with a
single puncture.)  Let $v$ be the vertex of the arrangement of~$\Gamma_1$
corresponding to this fundamental path. We test whether this fundamental
path is homotopic \emph{with fixed endpoints} to its counterpart in~$G_2$
\emph{on the punctured surface, except that the punctures on the endpoints
  of the path are removed}.

If these paths are not homotopic, clearly $G_1$ and~$G_2$ are not isotopic,
since any ambient isotopy between $G_1$ and~$G_2$ (fixing a subset of the
vertices) yields a homotopy of the fundamental path with fixed vertices.

On the other hand, assume that the images of the fundamental path in~$G_1$
and~$G_2$ are homotopic with fixed endpoints.  Assuming that the cycles are
homotopic in~$\Gamma_1$ and~$\Gamma_2$, we obtain by
Theorem~\ref{T:stablepunc} that some isotopy~$i$ maps each cycle
in~$\Gamma_1$ to the corresponding cycle in~$\Gamma_2$, not necessarily
pointwise, but preserving the orientations of the cycles.  As in the proof
of Proposition~\ref{P:fix}, $i$ can be viewed as an orientation-preserving
map automorphism of~$\Gamma_1$. We recall that $\Gamma'_1$ denotes the
arrangement of cycles of $\Gamma_1$ corresponding to a connected component
of $G$. To conclude, it suffices to prove that, for every connected
component of $G$, this map automorphism
restricts to the identity on~$\Gamma'_1$, and by connectedness
of~$\Gamma'_1$ it actually suffices to prove that it maps $v$ to itself.
But if $i$ were to map $v$ into another vertex, it would map a fundamental
path into a different one.

There just remains to prove that two different fundamental paths are
\textit{not} homotopic with fixed endpoints on this surface. Indeed, we can
perturb them slightly so as to make them disjoint (except at their
endpoints), so if they are homotopic, they bound a degree-two disk. If this
disk contains no other component of $G_1$, then $G_1$ has a face that is a
degree-two disk, which is excluded by the preprocessing step and the
hyperbolicity hypothesis. If it contains another component of $G_1$, this
component is contractible and without fixed vertices, which is again
excluded by the preprocessing.

Therefore $i$ is the identity, which concludes the proof of the counterpart
of Proposition~\ref{P:fix}.

\paragraph{End of Proof.} \ The components of $S \setminus \Gamma_1=S
\setminus \Gamma_2$ are now punctured disks instead of disks, but we made
sure that they only contain a \textit{single} puncture. Hence, since
Alexander's Lemma also applies to once-punctured disks \cite[Section
2.2.1]{fm-pmcg-11}, the end of the proof remains the same.

\paragraph{Conclusion.}  \ In summary, some minor modifications of our
algorithms can also handle the case where some vertices of the graph have
to be fixed during the entire isotopy.  In some cases, we need to use
homotopy tests for paths instead of cycles (which is a computationally
simpler problem, since it reduces to testing contractibility of cycles).
In particular, here are the counterparts of Theorems \ref{T:main-surf}
and~\ref{T:main-plane}:

\begin{theorem}\label{T:main-surf-punct}
  Let $S$ be an orientable surface, possibly with boundary.  Let $H$ be a
  fixed graph cellularly embedded on~$S$.  Let $G_1$ and~$G_2$ be two graph
  embeddings of the same graph~$G$ on~$S$, each in general position with
  respect to~$H$.  Furthermore, let $V_f$ be a set of vertices of~$G$.
  Given the combinatorial map of the arrangement of~$G_1$ with~$H$ (resp.,
  $G_2$ with~$H$), of complexity $k_1$ (resp.,~$k_2$), we can determine
  whether $G_1$ and~$G_2$ are isotopic, fixing the vertices in~$V_f$, in
  $O(k_1+k_2)$ time.
\end{theorem}
\begin{theorem}\label{T:main-plane-punct}
  Let $P$ be a set of $p$ points in the plane, and let $G_1$ and~$G_2$ be
  two piecewise-linear graph embeddings of the same graph~$G$ in
  $\R^2\setminus P$, of complexities (number of segments) $k_1$ and~$k_2$
  respectively.  Furthermore, let $V_f$ be a set of $q$~vertices of~$G$.
  We can determine whether $G_1$ and~$G_2$ are isotopic in $\R^2\setminus
  P$, fixing the vertices in~$V_f$, in time $O(n^{3/2}\log n)$ time, where
  $n$ is the total size of the input.  In more detail, the running time is,
  for any~$\varepsilon>0$,
\[
\begin{array}{ll}
    O\Big(&(k_1+p+q)\log(k_1+p+q)+(k_2+p+q)\log(k_2+p+q)\ + \\
    & \min\Big\{(k_1+k_2)(p+q), \quad  (p+q)^{1+\varepsilon}+(k_1+k_2)\sqrt{p+q}\log (p+q)\Big\}\Big).  \\
  \end{array}
\]
\end{theorem}

Moreover, the above proof implies that Corollaries \ref{C:ext-cont}
and~\ref{C:ext-pl} also extend to the case where some vertices of the graph
embeddings have to be fixed: any isotopy between $G_1$ and~$G_2$ extends to
an ambient isotopy, and that ambient isotopy can be chosen piecewise-linear
if $G_1$ and~$G_2$ are piecewise-linear.

Actually, if it were known that two piecewise-linear graph embeddings are
isotopic if and only if they are piecewise-linearly isotopic (with fixed
vertices), our result of Section \ref{S:stable} would extend rather
directly to the case of fixed surfaces, by the construction shown in
Figure~\ref{fig:fixvert}, where each vertex of~$G$ that has to be fixed is
replaced with a cycle of vertices and edges enclosing a boundary component
of the surface: Any piecewise-linear isotopy with fixed vertices in the
original graph corresponds to a piecewise-linear isotopy (without fixed
vertices) in the new graph and conversely, so we can apply our results on
graph isotopies without fixed vertices.  However, this fact seems harder to
prove for arbitrary continuous isotopies.

 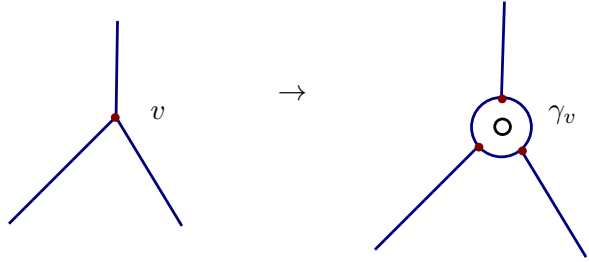
\begin{figure}
 \centering
 \def\svgwidth{8cm}
 \begingroup
  \makeatletter
  \providecommand\color[2][]{%
    \errmessage{(Inkscape) Color is used for the text in Inkscape, but the package 'color.sty' is not loaded}
    \renewcommand\color[2][]{}%
  }
  \providecommand\transparent[1]{%
    \errmessage{(Inkscape) Transparency is used (non-zero) for the text in Inkscape, but the package 'transparent.sty' is not loaded}
    \renewcommand\transparent[1]{}%
  }
  \providecommand\rotatebox[2]{#2}
  \ifx\svgwidth\undefined
    \setlength{\unitlength}{1856.77275391pt}
  \else
    \setlength{\unitlength}{\svgwidth}
  \fi
  \global\let\svgwidth\undefined
  \makeatother
  \begin{picture}(1,0.42472887)%
    \put(0,0){\includegraphics[width=\unitlength]{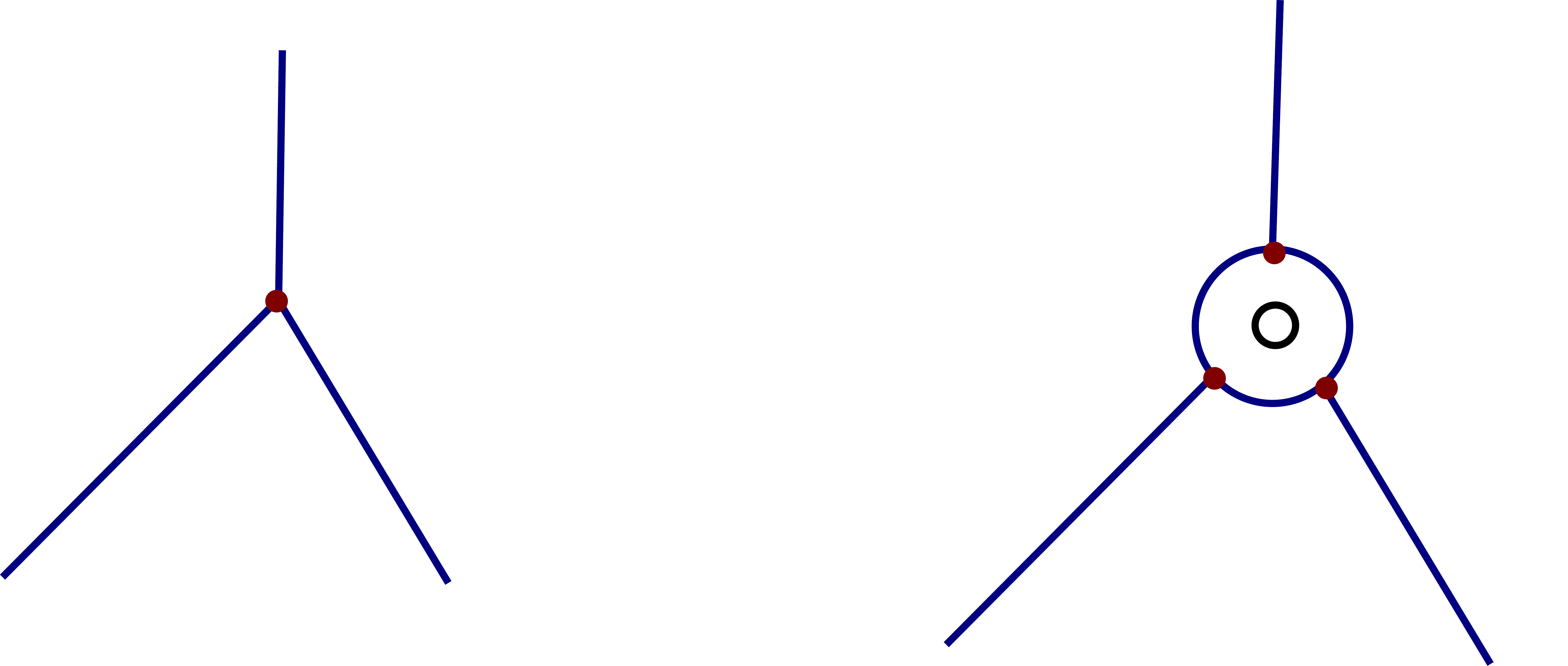}}%
    \put(0.44355378,0.2621616){\color[rgb]{0,0,0}\makebox(0,0)[lb]{\smash{$\rightarrow$}}}%
    \put(0.23428131,0.23015522){\color[rgb]{0,0,0}\makebox(0,0)[lb]{\smash{$v$}}}%
    \put(0.88918104,0.23138625){\color[rgb]{0,0,0}\makebox(0,0)[lb]{\smash{$\gamma_v$}}}%
  \end{picture}%
\endgroup
 \caption{Gadget reducing the case of fixed vertices to the general case: A
   fixed vertex~$v$ of degree~$d$ is replaced with $d$ vertices and edges
   forming a cycle~$\gamma_v$ around a new boundary component of the
   surface}
 \label{fig:fixvert}
 \end{figure}

%%%%%%%%%%%%%%%%%%%%%%%%%%%%%%%%%%%%%%%%%%%%%%%%%%%%%%%%%%%%%%%%%%%%%%%%%%%%%%%%%%

\section{Exceptional surfaces}\label{S:exc}

In this section, we prove that Theorem \ref{T:stable} also holds for
surface of nonnegative Euler characteristic. We note that in these cases we
obtain a stronger theorem that in the general case, as the ambient isotopy
we obtain maps the families of cycles pointwise. We split the proof in two
parts depending on whether the surface is a plane with boundaries in
Section~\ref{S:sphere} or a torus in Section~\ref{S:torus}. As a foreword,
we note that the results of Lemmas~\ref{lemeuler2},
Corollary~\ref{coreuler}, and Lemma~\ref{lemeuler} and the first result of
Proposition~\ref{proposition13} (no cycle in a stable family is
null-homotopic) also hold in the nonnegative Euler characteristic case.

\subsection{Sphere, disk, and annulus}\label{S:sphere}

Since the Euler characteristic of a surface is $\chi(S)=2-2g-b$, the only
cases where the surface is a plane with boundaries and has nonnegative
Euler characteristic are the sphere, the disk, and the annulus. In these
cases, the proof of Theorem \ref{T:stable} is a simple corollary of
Lemmas~\ref{lemeuler2} and~\ref{lemeuler}.  We even obtain a slightly
stronger statement, because we can take the isotopy to map \emph{pointwise}
$\Gamma_1$ to~$\Gamma_2$:
\begin{theorem}\label{T:stable2}
  Let $S$ be a sphere, a disk or an annulus and let
  $\Gamma_1=(\gamma_{1,1}, \ldots, \gamma_{1,n})$ and
  $\Gamma_2=(\gamma_{2,1}, \ldots \gamma_{2,n})$ be two \emph{stable}
  families of cycles on~$S$ in general position such that:
  \begin{enumerate}
  \item there exists an oriented homeomorphism~$h$ of $S$ mapping each
    cycle~$\gamma_{1,j}$ of~$\Gamma_1$ to the corresponding
    cycle~$\gamma_{2,j}$ of~$\Gamma_2$ not necesssarily pointwise, but
    preserving the orientations of the cycles, and
  \item each cycle of~$\Gamma_1$ is homotopic to the corresponding cycle
    of~$\Gamma_2$.
  \end{enumerate}
  Then there is an ambient isotopy of $S$ mapping each cycle of~$\Gamma_1$
  to the corresponding cycle of~$\Gamma_2$ pointwise.
\end{theorem}

\begin{proof}
  If the surface $S$ is a sphere or a disk, all the cycles in $\Gamma_1$
  and $\Gamma_2$ are null-homotopic, which is impossible as noted
  above. Hence these families are empty and the theorem is trivial.

  If the surface is an annulus, for $i=1,2$, we claim that there are no
  crossing points in $\Gamma_i$, i.e., all the cycles are simple and two
  distinct cycles do not intersect each other. Indeed, if there were a
  crossing point, the connected component of $\Gamma_i$ containing it would
  form a planar graph such that every vertex has degree four. Hence, by
  Corollary~\ref{coreuler}, there would be at least a $k$-gon with $k\le3$,
  contradicting, with Lemma~\ref{lemeuler}, the stability of the family
  $\Gamma_i$.

  Thus, $\Gamma_1$ is a family of disjoint simple cycles homotopic to the
  boundaries of the annulus, and the same holds for $\Gamma_2$. There is an
  isotopy of~$S$ mapping one family into the other if and only if they have
  the same ordering, as defined in the proof for the hyperbolic case. But
  this is exactly what the oriented homeomorphism between them
  ensures. This concludes the proof.
\end{proof}

%%%%%%%%%%%%%%%%%%%%%%%%%%%%%%%%%%%%%%%%%%%%%%%%%%%%%%%%%%%%%%%%%%%%%%%%%%
\subsection{Torus}\label{S:torus}

The proof in the case of the torus is slightly more involved. Let us
introduce a few definitions before delving into it. We choose a Euclidean
metric on the torus, which induces one on its universal cover
$\mathbb{R}^2$. This allows to define \emphdef{translations} on the torus,
which are projections of the usual translations of
$\mathbb{R}^2$. Geodesics on the torus lift to straight lines in the plane,
and two geodesics are homotopic if and only if these lines have the same
slope, as a slope $s=\frac{m}{n}$ determines a unique element $(m,n)$ with
$m \wedge n=1$ of the fundamental group of the torus. When we mention the
\emphdef{slope} of a geodesic on the torus, we refer to the slope of one of
its lifts in the universal cover. Note that as a translation is an
isometry, it maps a geodesic to another geodesic.

\begin{theorem}\label{T:stable3}
  Let $S$ be a torus and let $\Gamma_1=(\gamma_{1,1}, \ldots,
  \gamma_{1,n})$ and $\Gamma_2=(\gamma_{2,1}, \ldots \gamma_{2,n})$ be two
  \emph{stable} families of cycles on~$S$ in general position such that:
  \begin{enumerate}
  \item there exists an oriented homeomorphism~$h$ of $S$ mapping each
    cycle~$\gamma_{1,j}$ of~$\Gamma_1$ to the corresponding
    cycle~$\gamma_{2,j}$ of~$\Gamma_2$ not necesssarily pointwise, but
    preserving the orientations of the cycles, and
  \item each cycle of~$\Gamma_1$ is homotopic to the corresponding cycle
    of~$\Gamma_2$.
  \end{enumerate}
  Then there is an ambient isotopy of $S$ mapping each cycle of~$\Gamma_1$
  to the corresponding cycle of~$\Gamma_2$ pointwise.
\end{theorem}

\begin{proof}
  For all the cycles $\gamma$ in $\Gamma_1$ or $\Gamma_2$, we start by
  applying de Graaf and Schrijver \cite[Proposition~13]{gs-mcmcr-97} as in
  the proof of Proposition~\ref{proposition13}: Up to applying an isotopy
  of $S$, we can assume that $\gamma$ is contained in the
  $\varepsilon$-neighborhood of one of its corresponding geodesics (or of a
  point, if $\gamma$ is contractible).  As in the proof of
  Proposition~\ref{proposition13}, we infer that $\gamma$ is not
  contractible.  Since in a torus, geodesic cycles are either simple or
  multiple concatenations of the same simple cycle, their
  $\varepsilon$-neighborhoods are annuli. Hence, every cycle in $\Gamma_1$
  and $\Gamma_2$ can be assumed to lie in an annulus.

  If one of these cycles $\gamma \in \Gamma_i$ is non-simple, it forms a
  graph embedded on an annulus such that every vertex has degree four. By
  Corollary~\ref{coreuler}, one of the faces of this graph is a disk with
  degree lower than four, which with Lemma~\ref{lemeuler} contradicts the
  stability of $\Gamma_i$. Thus all the cycles in $\Gamma_1$ and $\Gamma_2$
  are simple. By the same argument, for $i=1,2$, two homotopic cycles in
  $\Gamma_i$ do not cross each other.

  The isotopy can then be found as follows. We split the proof in two
  cases.

  \textbf{Case 1:} If all the cycles in $\Gamma_1$ are homotopic or inverse
  homotopic\footnote{We say that $\gamma_1$ and $\gamma_2$ are inverse
    homotopic if $\gamma_1$ is homotopic to $\gamma_2^{-1}$.}, we just pick
  an arbitrary one, say $\gamma_{1,1}$, and apply the pointwise isotopy
  mapping it to $\gamma_{2,1}$, which exists because they are
  homotopic\footnote{Simple and noncontractible cycles which are homotopic
    are also isotopic, as proved by Eppstein~\cite[Theorem
    2.1]{e-c2mi-66}.}. Cutting the surface along
  $\gamma_{1,1}=\gamma_{2,1}$ gives an annulus. In this annulus, since all
  the other cycles in $\Gamma_1$ and $\Gamma_2$ are disjoint from these,
  the existence of an isotopy between them follows from the case of the
  annulus in Section~\ref{S:sphere}. After gluing back the boundaries
  together, this gives the desired isotopy of the torus.

  \textbf{Case 2:} If there are at least two homotopy classes (modulo
  inversion) in $\Gamma_1$, we pick two representatives, say $\gamma_{1,1}$
  and $\gamma_{1,2}$. As two couple of lines with the same slopes pairwise
  can be moved one to the other with a translation, by doing a translation
  of the torus, we can assume that $\gamma_{1,1}$ and $\gamma_{2,1}$ lie in
  the neighborhood of the same geodesic, as well as $\gamma_{1,2}$ and
  $\gamma_{2,2}$, and furthermore that the crossing points between
  $\gamma_{1,1}$ and $\gamma_{1,2}$ lie in a $\varepsilon$-neighborhood of
  the corresponding crossing points between $\gamma_{2,1}$ and
  $\gamma_{2,2}$\footnote{This is not necessarily the case a priori, since
    a given crossing point $c$ between $\gamma_{1,1}$ and $\gamma_{2,1}$
    can be matched to another crossing point than $h(c)$. Note that this is
    why the result in the torus case in stronger than in the general case,
    in which the crossing points can \emph{not} necessarily be matched.}.

  We are now in the same situation as in Section~\ref{S:follow}: Since both
  couples of cycles $\gamma_{1,1}$ and $\gamma_{2,1}$ and $\gamma_{1,2}$
  and $\gamma_{2,2}$ lie in the $\varepsilon$-neighborhood of the same
  geodesic, if we take as stable families $\Gamma'_1=\gamma_{1,1}\cup
  \gamma_{1,2}$ and $\Gamma'_2=\gamma_{2,1} \cup \gamma_{2,2}$, we can
  similarly define corridors, as well as edge and vertex polygons. Then
  Proposition~\ref{P:trgle} holds with exactly the same proof. Since there
  is only one cycle of $\Gamma'_1$ in each corridor, Lemma~\ref{L:homeo}
  also holds trivially.  Hence by applying the same techniques, we can
  conclude that there exists an ambient isotopy mapping $\Gamma'_1$ to
  $\Gamma'_2$. Note that here, since the crossing points of the cycles in
  $\Gamma'_1$ have been matched, the isotopy we obtain is also pointwise.

  Finally, cutting along these cycles cuts the surface into one or more
  disks, and the isotopy between $\Gamma_1$ and $\Gamma_2$ is obtained by
  applying Alexander's lemma separately on each of these disks.
\end{proof}

%%%%%%%%%%%%%%%%%%%%%%%%%%%%%%%%%%%%%%%%%%%%%%%%%%%%%%%%%%%%%%%%%%%%%%%%%%

\section{Conclusion}

We have given an optimal algorithm to test whether two graph embeddings of
the same graph on an orientable surface are isotopic, where the vertices
are allowed to move.  On the other hand, the case of non-orientable
surfaces remains open.  Our characterization heavily relies on testing the
existence of an oriented homeomorphism; it is not clear how to adapt this
test in the non-orientable case.

Regarding the case of the punctured plane, the most obvious open question
is to improve the running time of the algorithm.
Bespamyatnikh~\cite[Theorem~7]{b-ehpp-04} describes an
$O(n^{4/3}\*\polylog(n))$-time algorithm for testing path homotopy;
however, it is not clear that this algorithm extends to the homotopy test
for cycles.  One obstacle for this extension is that our cycles are not
simple and may make turns always in the same direction: the algorithm by
Bespamyatnikh~\cite[Section~5.5]{b-ehpp-04} considers maximal subpaths that
always turn in the same direction, but in our case such maximal subpaths
may be cycles without ``starting'' and ``ending'' points, for which the
same approach does not seem to work.

Also, we only test the existence of a topological isotopy: The edges are
allowed to bend during the deformation.  It is easy to see that, in the
presence of obstacles, the existence of a topological isotopy between two
straight-line embeddings does not imply the existence of a straight-line
isotopy, in contrast to the case without
obstacle~\cite{gs-gifpm-01,bs-lie-78}.  Could it be that, in such a
situation, there exists a straight-line isotopy after splitting each edge
in two (or a constant number of) segments?  Computing such an isotopy
efficiently may be a not easy task, but related
techniques~\cite{lp-mpgdb-08} might apply.

Finally, in both the surface model and the punctured plane model, computing
shortest graph embeddings within a given isotopy class would be very
interesting, and would generalize known results for computing shortest
paths within a given homotopy
class~\cite{hs-cmlpg-94,ekl-chspe-06,b-chspp-03,cl-oslos-05,cl-opdsh-07,ce-tnpcs-10},
even though we expect the problem to be much harder.

%%%%%%%%%%%%%%%%%%%%%%%%%%%%%%%%%%%%%%%%%%%%%%%%%%%%%%%%%%%%%%%%%%%%%%%%%%
\subsection*{Acknowledgment}

We thank Francis Lazarus and Julien Rivaud for sending us an advance copy
of their paper~\cite{lr-hts-12}, which inspired us for the construction of
the stable family in Section~\ref{S:ladeg}, and the anonymous referees for
helpful comments.

\bibliographystyle{siam}

%%%%%%%%%%%%%%%%%%%%%%%%%%%%%%%%%%%%%%%%%%%%%%%%%%%%%%%%%%%%%%%%%%%%%%%%%%
%%%%%%%%%%%%%%%%%%%%%%%%%%%%%%%%%%%%%%%%%%%%%%%%%%%%%%%%%%%%%%%%%%%%%%%%%%
%%%%%%%%%%%%%%%%%%%%%%%%%%%%%%%%%%%%%%%%%%%%%%%%%%%%%%%%%%%%%%%%%%%%%%%%%%

\appendix\normalsize

\end{document}